\documentclass[final,3p]{elsarticle}

\usepackage{amssymb,amsmath,amsthm}
\usepackage{mathrsfs}
\usepackage{xcolor}
\usepackage{hyperref}
\usepackage{bookmark}
\usepackage{xspace}
\usepackage{mathtools}
\usepackage[all]{xy}
\usepackage{extarrows}
\usepackage{pgfplots}

\usepackage{tikz}
\usetikzlibrary{positioning}
\usetikzlibrary{arrows}
\usetikzlibrary{petri}
\usepackage{caption}

\usepackage{amsfonts}
\usepackage[T1]{fontenc}
\usepackage{tikz}
\usetikzlibrary{arrows,arrows.meta,automata,positioning} 

\newtheorem{theorem}{Theorem}

\newtheorem{lemma}[theorem]{Lemma}
\newtheorem{definition}[theorem]{Definition}
\newtheorem{remark}[theorem]{Remark}

\newcommand{\Wlog}{W.l.o.g.\xspace}
\newcommand{\mywlog}{w.l.o.g.\xspace}

\newcommand{\Zrun}{$\Z$-run\xspace}
\newcommand{\Zruns}{$\Z$-runs\xspace}
\newcommand{\Fair}{Fair\xspace}
\newcommand{\fair}{fair\xspace}

\newcommand{\balrun}{balanced\xspace}
\newcommand{\balrunable}{balanceable\xspace}

\newcommand{\intcone}[1]{\text{\sc IntCone}(#1)}
\newcommand{\diff}[2]{\Delta_{#1}(#2)}
\newcommand{\wrprod}[2]{#1\wr #2}

\newcommand{\symg}[1]{S_{#1}}
\newcommand{\altg}[1]{A_{#1}}
\newcommand{\cycg}[1]{Z_{#1}}
\newcommand{\tg}[1]{I_{#1}}
\newcommand{\const}[1]{\overline{#1}}

\newcommand{\absv}[1]{|#1|}
\newcommand{\OO}{{\cal O}}
\newcommand{\Th}{\Theta}
\newcommand{\Npos}{\N_+}

\newcommand{\svass}[1]{#1-{\sc Vass}}
\newcommand{\vect}[1]{\mathbf{#1}}
\newcommand{\vr}[1]{\vect{#1}}

\newcommand{\vass}{\text{\sc vass}\xspace}

\newcommand{\parvass}[1]{{$#1$-\vass}\xspace}

\newcommand{\runarrow}{\longrightarrow}

\newcommand{\norm}[1]{\Vert #1 \Vert}

\newcommand{\NP}{\text{\sc NP}\xspace}
\newcommand{\NL}{\text{\sc NL}\xspace}
\newcommand{\tower}{\text{\sc Tower}\xspace}
\newcommand{\FF}[1]{\mathcal{F}_{#1}}
\newcommand{\pspace}{\text{\sc PSpace}\xspace}

\newcommand{\pair}[2]{\langle #1,#2 \rangle}
\newcommand{\tuple}[1]{\langle #1 \rangle}
\newcommand{\wrtuple}[2]{(#1)\wr #2}

\newcommand{\trans}[1]{\stackrel{#1}{\longrightarrow}}

\newcommand{\reachprob}[1]{$#1$-\vass-{\sc{Reach}}}

\newcommand{\id}[1]{\textrm{Id}_{#1}}

\newcommand{\size}[1]{|#1|}

\newcommand{\Z}{\mathbb{Z}}

\newcommand{\N}{\mathbb{N}}

\newcommand{\B}{\mathcal{B}}

\newcommand{\V}{\mathcal{V}}

\newcommand{\rev}[1]{{#1}^{\text{\tiny rev}}}

\newcommand{\set}[1]{\left\{#1\right\}}

\newcommand{\setof}[2]{\left\{#1 \, \middle\vert \, #2\right\}}

\newcommand{\subseteqfin}{\subseteq_\text{fin}}

\newcommand{\setfromto}[2]{\set{#1,\ldots , #2}}
\newcommand{\setto}[1]{[#1]}

\newtheorem{example}{Example}

\newcommand{\para}[1]{\subparagraph*{\rm \bf #1.}}

\journal{Information and Computation}

\begin{document}

\begin{frontmatter}



\title{Reachability in Symmetric VASS}


\author[uw]{{\L}ukasz Kami{\'n}ski\fnref{s1}}
\ead{l.kaminski5@uw.edu.pl}
\author[uw]{S{\l}awomir Lasota\fnref{s2}}
\ead{sl@mimuw.edu.pl}


\fntext[s1]{Partially supported by the NCN grant 2024/55/B/ST6/01674.}
\fntext[s2]{Partially supported by the ERC grant INFSYS, agreement no. 950398, and the NCN grant 2024/55/B/ST6/01674.}

\affiliation[uw]{organization={University of Warsaw},
             country={Poland}}

\begin{abstract}
  We investigate the reachability problem in \emph{symmetric} vector addition systems with states (\vass),
  where transitions are invariant under a group of permutations of coordinates.
  One extremal case, the trivial groups, yields general \vass.
  In another extremal case, the symmetric groups,
  we show that the reachability problem can be solved in \pspace, 
  regardless of the dimension of input \vass
  (to be contrasted with Ackermannian complexity in general \vass).
  We also consider other groups, in particular alternating and cyclic ones.
  Furthermore, motivated by the open status of the reachability problem in data \vass,
  we estimate the gain in complexity when the group arises as a combination of the trivial and symmetric groups. 
\end{abstract}




\end{frontmatter}




\section{Introduction}

Petri Nets, equivalently presentable as vector addition systems with states (\vass), 
are an established model of concurrency with widespread applications.
The central algorithmic problem for this model
is the \emph{reachability problem} which asks whether from a given initial configuration 
there exists a sequence of valid execution steps reaching a given final configuration.
The decidability of the problem was established in 1981 by Mayr \cite{Mayr81}, 
and subsequently improved by Kosaraju \cite{Kosaraju82} and Lambert \cite{Lambert92}.
An exponential space lower bound was shown by Lipton already in 1976~\cite{Lipton76}.
For around 40 years this was the only known complexity bound, 
and the complexity of the problem remained for a long time
as one of the hardest open questions in the verification of concurrent systems.
Only in the last few years we have seen a sudden progress, namely
an Ackermannian upper bound~\cite{LS19},
a breakthrough non-elementary lower bound
~\cite{CLLLM19,jacm}, and
finally the matching Ackermannian lower bound
(independently \cite{CO21} and \cite{L21}). 

All the above-mentioned bounds are more fine-grained and apply to 
the problem pa\-ra\-me\-trised by the dimension,
where the input is restricted to
$d$-dimensional \vass (\parvass d) for fixed $d \in \N$.
Both the upper and lower bounds have been subsequently further improved \cite{Las22,COLO23,FYZ24},
and currently the reachability problem for \parvass{d} is known to belong to
the complexity class $\FF d$, and
for \parvass{(2d+3)} it is known to be $\FF d$-hard.
Here $\FF d$ denotes the $d$th level of the hierarchy
of complexity classes
corresponding to Grzegorczyk's fast growing function hierarchy.
For instance, the currently best lower bound of \cite{COLO23} shows the reachability problem for 
\parvass 9 to be hard for the class $\FF 3 = \tower$. 
Subsequently, \cite{CO22} proved the same complexity bound even for \parvass 8.
Summing up, for sufficiently large dimensions the lower bound
is prohibitive.

\para{Symmetric \vass}

In this paper we investigate subclasses of \parvass d which are \emph{symmetric},
meaning their sets of transitions are invariant under certain permutations of the coordinates $\{1,\ldots, d\}$.
Our study is parametric in the choice of a group of permutations $G \leq \symg d$
of coordinates, and our objective is to analyze the gain
in complexity to be
achieved when the input to the reachability problem is restricted to 
\parvass G, the subclass of those \parvass d whose sets of transitions are invariant under
permutations $\sigma \in G$.
In one extreme case, when $G$ is the trivial permutation group $\tg d$ containing just the identity
permutation, \parvass G are just general \parvass d.
In another extreme case, when $G=\symg d$ is the symmetric group of degree $d$,
transitions of \parvass G are invariant under all permutations of coordinates.
The two extreme cases may be combined. 
For instance, one can consider the permutation group
$G = \wrprod {\tg n} {\symg d} \leq \symg {nd}$ of degree $nd$
containing all permutations that are identity inside each of the $n$-element blocks:
\[
\{1, \ldots, n\}, \quad
\{n+1, \ldots, 2n\}, \quad
\ldots, \quad
\{n(d-1) + 1, \ldots, nd\},
\]
but permute arbitrarily the whole blocks.

Our principal motivation is to exploit symmetry of a model in order to
lower the prohibitively high complexity of the reachability problem in the general case.
Another motivation comes from the model of \emph{data \vass} \cite{Lasota16,LNORW07,RF11}
 extending plain \vass with
data.
In terms of Petri Nets, the extension allows tokens to carry data values coming from
some infinite countable set, and allows transitions to test (dis)equalities between
data values of involved tokens.
When $G = \wrprod {\tg n} {\symg d} \leq \symg {nd}$, the \parvass G
are exactly data \vass of dimension $n$, where the set of data values is 
restricted to be finite%
\footnote{
For the other way around, $n$-dimensional 
data \vass can be defined as \parvass G, where $G = \wrprod{\tg n}{\symg \infty}$
combines the trivial group $\tg n$ with the symmetric group $\symg \infty$ of infinite countable degree.
},
of size $d$.
It is not known if the reachability problem is decidable for data \vass, and our current study
may shed some more light on this hard open problem.

Symmetric \vass should not be confused with Symmetric Nets \cite{HKPT06}. 
While both formalisms exploit symmetries, they do so in different ways.
Symmetric Nets are a high-level extension of Petri Nets.

\para{Contribution}
Intuitively speaking,
the larger the group $G$, the lower the complexity of the reachability problem, since
whenever $G\leq H$, every \parvass H is automatically a \parvass G.
The main contribution of this paper is twofold.
First, we concentrate on the potentially easiest cases, namely symmetric groups $\symg d$, and 
discover a huge complexity drop:
regardless of the dimension $d\geq 2$, the reachability problem is \pspace-complete
for \vass invariant under symmetric groups.%
\footnote{
In dimension 1 the permutation group $G$ is irrelevant, and the problem is \NL-complete 
or \NP-complete, 
for input represented in unary or binary, respectively \cite{HKOW09}.
}
We also prove the same complexity for another potentially easy case, namely for the alternating groups
$\altg d$:
again, regardless of the dimension $d\geq 3$, the reachability problem is \pspace-complete
for \vass invariant under alternating groups.
Both the \pspace decision procedures are designed for a fixed dimension $d$, but work equally
well when $d$ is part of input.
We thus get a uniform \pspace upper bound for \vass invariant under symmetric (resp.~alternating) groups.
On the other side, we prove that cyclic groups $\cycg d$ are hard, namely the reachability problem
for \parvass d invariant under 
$\cycg d$ is as hard as in \parvass{\Th(d)}.

The case of trivial group $G=\tg d$ coincides with general \parvass d.
As our second main contribution we investigate combinations of symmetric groups and trivial groups,
providing one positive and one negative result.
On the one hand, when $G = \wrprod {\tg n} {\symg d} \leq \symg {nd}$
(as discussed above) and $n\geq 2$, the reachability problem
is as hard as in the case of \parvass {(n-1)d}.
From the perspective of data \vass
this may be interpreted as a bad news: the complexity grows with the increasing
number $d$ of data values.
On the other hand, we investigate another combination of the two groups, namely
$G = \wrprod {\symg d} {\tg n} \leq \symg {dn}$ containing all permutations that
independently and arbitrarily permute each of the $d$-element blocks: 
\[
\{1, \ldots, d\}, \quad
\{d+1, \ldots, 2d\}, \quad
\ldots, \quad
\{d(n-1) + 1, \ldots, dn\},
\]
but preserve each of the blocks.
In this case we provide an exponential-time reduction
to \parvass n, i.e., this time
the complexity is independent of the degree $d$ of the symmetric group.
The complexity in this case is thus (significantly) lower than for $\wrprod {\tg n} {\symg d}$.
This is in agreement with 
$\wrprod {\tg n} {\symg d}$ being a subgroup of $\wrprod {\symg d} {\tg n}$,
up to isomorphism of permutation groups.
Using similar methods we also prove that the reachability problem for a subclass
of data \vass invariant under independent
data permutations in each dimension
(definable as  \parvass{(\wrprod{\symg \infty}{\tg n})} in our setting)
is decidable.

\medskip

\para{Overview}
We start in Section \ref{sec:prelim} with preliminaries on \vass and permutation groups.
In Section \ref{sec:lowerbound} we prove \pspace-hardness of the reachability problem
for \parvass G, for any $G \leq \symg d$ of dimension $d \geq 2$.
In Section \ref{sec:transitive} we introduce and study \emph{fair} permutation groups,
one of the key technical concepts of the paper, and
prove a general \pspace upper bound for such groups.
Then in Section \ref{sec:symg} we prove that symmetric groups are fair,
thus obtaining the \pspace upper bound for \parvass {\symg d}.
In Section \ref{sec:comb} we study combinations of symmetric and trivial groups,
proving both positive and negative complexity results.
In Section \ref{sec:data} we apply the developed techniques to data \vass,
proving decidability of the reachability problem for a subclass of data \vass.
In Section \ref{sec:other} we strengthen the \pspace upper bound to alternating groups
by showing that they are also fair, 
and prove hardness of cyclic groups.
We conclude in Section \ref{sec:final}.

This article is an extended version of the conference paper \cite{KL25}. 
The main addition is the content presented in Section \ref{sec:data}. 
We have also corrected several minor errors.


\section{Symmetric \vass}
\label{sec:prelim}

As usual, let $\Z, \N$ denote integers and nonnegative integers, respectively.
The value of the $i$th coordinate (called also the $i$th \emph{counter}) of a vector $\vr w\in\Z^d$ is written as 
$\vr w(i)$, namely
$\vr w = (\vr w(1), \ldots, \vr w(d))$, i.e.,
we identify a vector with a function $\vr w \colon \setto d\to\Z$,
where $\setto d = \{1, \ldots, d\}$.
The \emph{norm} of a vector $\vr v \in \Z^d$ is  the maximum of absolute values of its coordinates,
i.e., $\norm{\vr v} = \max_{i \in \setto d} \absv{\vr v(i)}$.
Then the \emph{norm} of a set of vectors $X$ is defined as $\norm{X} = \max_{\vr v \in X} \norm{\vr v}$.
For $n\in\Z$, by $\const n$ we denote the constant vector $\const n = (n, \ldots, n) \in \Z^d$,
in the dimension $d$ to be always determined by the context.
For a finite set $X$,
by $\size X\in \N$ we denote its size.

\para{Vector addition systems with states}

For $d\in\Npos=\N\setminus \{0\}$,
we define $d$-dimensional \emph{vector addition systems with states}, denoted shortly as 
\parvass d, or simply \vass when the dimension $d$ is irrelevant or clear from the context.
A \parvass d $\V=(Q,T)$ consists of
a finite set $Q$ of states and a finite set of transitions $T\subseteq Q\times \Z^d \times Q$.
A \emph{configuration} $c$ of $\V$ consists of a state $q\in Q$ and a nonnegative vector $\vr w\in\N^d$,
and is written as $c=q(\vr w)$.
A transition $t=(q,\vr v,q')$ induces \emph{steps} 
\begin{align} \label{eq:step}
q(\vr w) \trans{t} q'(\vr w')
\end{align}
between configurations, where $\vr w' = \vr w+\vr v$.
We refer to the vector $\vr v\in\Z^d$ as the \emph{effect} of the transition $(q,\vr v,q')$ 
or of an induced step.
The \emph{norm} of a transition $t = (q, \vr v, q')$ is defined as $\norm{t} := \norm{\vr v}$.
Similarly, for a configuration $c = q(\vr w)$, we define $\norm{c} := \norm{\vr w}$.
The \emph{norm} of a \vass $\V$ is the maximal norm of its transitions, i.e.,
$\norm{\V} := \max_{t \in T} \norm{t}$.

A \emph{run} $\pi$ in $\V$ is a sequence of steps with the proviso that the target configuration of every step matches
the source configuration of the next one:
\begin{align} \label{eq:path}
\pi \ = \ c_0 \trans{t_1} c_1 \trans{} \ldots \trans{t_n} c_n.
\end{align}
We say that the run $\pi$ is \emph{from} $c_0$ \emph{to} $c_n$, call $c_0, c_n$ the \emph{source} and the \emph{target} 
configurations of the run, respectively, and write $c_0\trans{\pi}c_n$.
When the source configuration $c_0$ is clear from the context, we may 
identify a run \eqref{eq:path} from $c_0$ with the sequence $t_1 \ldots t_n$ of transitions fired.
The \emph{effect} 
of a run is the sum of effects of all steps, and its \emph{length} is the number $n$ of steps.
We also write $c \trans{} c'$ if there is some run from $c$ to $c'$.
We often use the natural operation of concatenation of runs:
if $\alpha \colon c\trans{} c'$ and $\beta \colon c' \trans{} c''$ then $(\alpha ; \beta) \colon c \trans{} c''$, 
assuming that
the target configuration of the former run coincides with the source configuration of the latter one.
For convenience, we allow ourselves to drop this assumption and use the convention that 
the latter run $\beta$ is implicitly shifted so that its source matches the target of $\alpha$, as long as
the shifted $\beta$ is still a run.
Under this convention, we write $\alpha^n$ for the $n$-fold concatenation of $\alpha$.

In the sequel we refer to the \emph{state graph} of a \vass $\V=(Q, T)$,
to mean a directed graph $(Q, E)$
whose nodes are states $Q$, and edges $E\subseteq Q^2$ are those pairs $(q, q')$
for which there is a transition $(q, \vr v, q')\in T$, for some $\vr v\in\Z^d$.
This graph may contain self-loops.

\para{Reachability problem}

One of the most fundamental computational problems studied in the setting of \vass is
the \emph{reachability problem}
(for the sake of this paper we prefer to formulate the problem for a \emph{fixed} dimension $d$):

\begin{quote}
\reachprob d:\\
Given a \parvass d $\V$ together with two configurations, source $s$ and target $t$,
one asks if $\V$ has a run from $s$ to $t$.
\end{quote}
Complexity of this problem in subclasses of \emph{symmetric} \vass, to be defined below, 
is the main topic studied in this paper.
An input to the problem is a triple $(\V, s, t)$.
For simplicity, 
such a triple $(\V, s, t)$, or a pair $(\V, s)$ with just the source state, we also call a \vass,
hoping that this does not lead to confusion.
We also use shorthands $\norm{\V, s} := \max\{\norm{\V}, \norm s\}$ and
$\norm{\V, s, t} := \max\{\norm{\V}, \norm s, \norm t\}$.

\para{Symmetric \vass}

Let $\symg d$ denote the symmetric group containing all permutations of $\setto d$.
Permutations of $\setto d$ act on vectors $\vr w\in\Z^d$ by permuting dimensions,
namely  $\sigma\in\symg d$ maps $\vr w$ to the vector $\sigma(\vr w)$ defined by
$\sigma(\vr w)(i) = \vr w(\sigma^{-1}(i))$, or equivalently $\sigma(\vr w)(\sigma(i)) = \vr w(i)$.
Treating a vector as a function $\vr w \colon \setto d \to \Z$, we may write
\[
\sigma(\vr w) = \vr w \circ \sigma^{-1}, \quad \text{ or equivalently } \quad \sigma(\vr w) \circ \sigma = \vr w.
\]
We naturally extend this action to transitions, $t=(q, \vr v, q') \longmapsto \sigma(t)=(q, \sigma(\vr v), q')$, 
and likewise to steps, runs, etc., as expected.
Given a group of permutations $G\leq \symg d$, by a \parvass G, we mean any \parvass d $\V = (Q, T)$
whose transitions are invariant under the action of $G$, namely
for every permutation $\sigma \in G$ and transition $t\in T$ we have $\sigma(t) \in T$.
Thus, for every $G\leq\symg d$, we identify a subclass of \parvass d admitting the invariance property.
In the most restrictive case of $G=\symg d$, we get the subclass of \parvass d invariant under
all permutations of dimensions.
In the less restrictive  case of trivial group $\tg d = \{\id d\}$, we get
the whole class of all \parvass d.
The name \emph{symmetric \vass} is used for 
\parvass G regardless of the group $G$.

\para{Reachability problem for symmetric \vass}

In our investigations we always assume that a group $G\leq \symg d$ (and hence also the dimension of \vass)
is fixed.%
\footnote{
Nevertheless, all our subsequent \pspace decision procedures are uniform with respect to dimension, namely they work 
within the same complexity bounds if dimension $d$ is part of the input.
}
The \vass reachability problem, when its input is restricted to \parvass G for a fixed 
permutation group $G$, we call \reachprob G.
For the trivial group $G = \tg{d} = \{\id d\}$ we get
\reachprob d.

The set of transitions of a \parvass G is determined uniquely by representatives of orbits, i.e.,
by representatives of the equivalence relation on transitions:
$t\equiv_G t'$ if $t=\sigma(t')$ for some $\sigma \in G$.
Unfolding the action of $G$ on transitions, we have that,
$(q, \vr v, p) \equiv_G (q', \vr v', p')$ if and only if $q=q'$, $p=p'$, and $\vr v = \sigma(\vr v')$ for some
$\sigma\in G$.
Equivalence classes of $\equiv_G$ are called \emph{$G$-orbits} of transitions, or simply
\emph{orbits} when the group is clear from the context.
When measuring the size of the input
we assume a succinct representation, where
transitions $T$  of a \parvass G are given by orbit representatives, i.e.,
by a subset $\widetilde T\subseteq T$ containing one transition from every orbit.
By the \emph{size} of such a representation $(Q, \widetilde T)$ of a \parvass G $\V$ 
we mean the bitsize of its description, e.g.,
\[
\size {\V} \ = \ \size{\widetilde T} \cdot \big(2\cdot{\size Q} + d \cdot 
\lfloor \log (2\cdot\norm{\widetilde T}+1)\rfloor  \big),
\]
where $\norm{\widetilde T} = \max_{t \in \widetilde T} \norm{t}$.
Note that $\size {\V}$ is independent of a choice of orbit representatives, and
logarithmic in terms of $\norm {\widetilde T}$.


Whenever $G'\leq G\leq \symg d$, i.e., $G'$ is a permutation subgroup of $G$,
(a representation of) a \parvass G may be transformed into 
(a representation of) a \parvass {G'} in two different ways.
Let $\V = (Q, T)$ be a  \parvass G. 
On the one hand, $\V$ is itself a \parvass {G'}, as its transitions are automatically invariant under 
the action of $G'$,
but the representation of $\V$ seen as a \parvass {G'} may be exponentially larger than its representation
when seen as \parvass G
(since the number of $G'$-orbits of $T$ 
may be exponentially larger than the number of $G$-orbits of $T$). 
This yields an exponential (but not necessarily polynomial) reduction from \reachprob G to \reachprob {G'}.
On the other hand, a \emph{representation} of $\V$ is automatically a representation of 
another \parvass {G'} 
$\V'=(Q, T')$,
where the set of transitions $T'\subseteq T$ may be of exponentially smaller size than $T$,
but the number of $G'$-orbits of $T'$ is the same as the number of $G$-orbits of $T$.
Clearly, every run 
of $\V'$ is a run 
of $\V$, but the converse does not hold in general.

\para{$\Z$-reachability}

An important relaxation of reachability is \emph{$\Z$-reacha\-bi\-lity}, where one considers a relaxed notion
of configurations, namely \emph{$\Z$-configurations} $Q\times\Z^d$ instead of $Q\times\N^d$. 
Contrary to configurations, values in $\Z$-configurations may drop below 0.
We define \emph{\Zruns} exactly like runs, but using $\Z$-con\-fi\-gu\-ra\-tions instead of configurations.
Each \Zrun, and hence also each run, induces a unique path in the state graph.
Conversely, every path in the state graph gives rise to multiple different \Zruns,
depending on the choice of source vector, and of transitions witnessing consecutive edges of a path.

The \emph{$\Z$-reachability problem} asks, given a \parvass G $(\V,s,t)$, if $\V$ has a \Zrun
from $s$ to $t$.
We prove that the problem lies in \NP,
regardless of the dimension $d$ of \vass and of the group $G\leq \symg d$.
The same argument works
even in the uniform version of the problem, 
when the group $G$ is given as a part of the input (in any reasonable representation).
This observation will be used later in
the proofs of the main results.

\begin{lemma}\label{lem:zrun}
    The $\Z$-reachability problem for \parvass G 
    is \NP-complete, for every $G\leq \symg d$.
\end{lemma}

The lemma does not follow immediately from the \NP upper bound for
$\Z$-reachability in \vass,
as the input is represented more succinctly in our case.
In the proof we use the integer analogue of Caratheodory's theorem \cite{IntCaratheodory} to
get a small solution property: if there is a $\Z$-run between two given configurations
then there is one using only a subset of transitions of polynomial size.
By guessing this subset we reduce to the $\Z$-reachability 
problem in $\vass$, which is in \NP.
%

\begin{proof}[Proof of Lemma \ref{lem:zrun}]
    \NP-hardness, already in dimension 1, is shown similarly as \NP-hardness of reachability \cite{HKOW09},
    by reduction from \textsc{subset-sum}.
    We thus concentrate on \NP-membership.

We use a result of \cite{IntCaratheodory} stating that a positive integer combination
of a set of vectors  $X\subseteqfin \Z^d$ is always presentable as such a combination of 
a subset of $X$ of size logarithmic in $\norm X$.
Formally, we define the \emph{integer cone} of a finite set $X \subseteqfin \Z^d$ as
\[
\intcone{X} = \setof{\lambda_1 x_1 + \ldots + \lambda_t x_t}{t\geq 0 ; \ 
x_1, \ldots, x_t \in X ; \ \lambda_1, \ldots, \lambda_t \in \Npos}.
\]

\begin{theorem}[\cite{IntCaratheodory}, Thm.~1(ii)]\label{thm:cara}
    Let $X \subseteqfin \Z^d$ and 
    $\vr b \in \intcone{X}$.
    Then $\vr b \in \intcone{\tilde{X}}$ for some $\tilde{X} \subseteq X$ of size
    $\size{\tilde{X}} \leq 2d \cdot \log(4d \cdot \norm X)$.
\end{theorem}

    Assume an input  \parvass G $(\V, s, t)$, where
    $\V = (Q, T)$, $s = q(\vr v)$ and $t = p(\vr w)$, represented by 
    orbit representatives $\tilde T \subseteq T$.
    Importantly, the number of edges $\size E$ of the state graph $(Q, E)$ of $\V$ 
    is at most $\size{Q}^2$ and (at most linearly) depends
    on $\size{\tilde T}$, but not on $\size T$.
    We design a nondeterministic polynomial time procedure to decide if $\V$ has a \Zrun
    from $s$ to $t$.

    Every subset $E'\subseteq E$ induces a subgraph of the state graph.
    The set $E'$ is called \emph{connected} if the induced subgraph is connected 
    when orientations of edges $E$ are ignored.  
    As discussed above, a \Zrun from $s$ to $t$ induces a path 
    from $q$ to $p$ in the state graph.
    We say that the \Zrun \emph{uses} the edge $e$ if $e$ appears in the induced path.
    Clearly, the set of edges used by any \Zrun is connected.
    We rely on this fact, and therefore the decision procedure,
    as the first step,  guesses
    a connected subset $E'$ of $E$.
    Since $E'$ is guessed arbitrarily,
    the remaining part of the decision procedure is devoted to checking if  
    there is a \Zrun that induces a path from 
    $q$ to $p$ and uses exactly all edges of a given $E'$ (we call such \Zruns \emph{acceptable}), 
    and whose effect is $\vr w - \vr v$ ($*$).
    
    By \emph{Parikh image} of a \Zrun $\pi$ we mean a function $P_{\pi} \colon T\to\N$ that assigns to
    every transition $t\in T$ the number of appearances of $t$ in $\pi$.
    Let $T'\subseteq T$ denote the subset of those transitions that induce edges from $E'$:
    \[
    T' = \setof{(q', \vr x, p')\in T}{(q', p')\in E'}.
    \]
    As the second step, the decision procedure guesses an arbitrary acceptable \Zrun  which is 
    Parikh-minimal (minimal with respect to pointwise order on Parikh images). 
    The length of a guessed \Zrun (and also of the induced path) may be bounded polynomially.
    Indeed, it is enough to guess a \Zrun inducing a simple path from $q$ to $p$, 
    plus a finite set of \Zruns inducing simple cycles that (jointly) use every edge $e\in E'$. 
    Since $E'$ is connected, the guessed \Zruns may be concatenated into a single acceptable \Zrun $\pi$,
    of length at most $\size Q \cdot \size {E'}$.
    Let $\vr e  \in\Z^d$ denote the effect of this guessed \Zrun, and $\vr p = P_{\pi} \in \N^{T'}$ its Parikh image
    (being acceptable, the \Zrun only uses transitions from $T'$).
    As the above-described guessing procedure includes all Parikh-minimal acceptable \Zruns, 
    our task, namely checking whether
    
    \begin{itemize}
    \item[($*$)] there is an acceptable \Zrun whose effect is $\vr w - \vr v$,
    \end{itemize}
     
    \noindent
    reduces to checking whether
    
    \begin{itemize}
    \item[($\circ$)] there is an acceptable \Zrun whose effect is $\vr w - \vr v$ 
    and whose Parikh image dominates $\vr p$ with respect to the pointwise ordering.
    \end{itemize}
    
    \noindent
    Let $X\subseteqfin\Z^d$ be the set of effects of all simple cycles that use only transitions from $T'$.
    Note that $\norm{X} \leq \norm{T'} \cdot \size{Q}\leq \norm{T} \cdot \size{Q}$.
    We claim that ($\circ$) is equivalent to
    \begin{align} \label{eq:intconeX}
    \vr w - \vr v - \vr e \in \intcone{X}.
    \end{align}
    Indeed, acceptable \Zruns satisfying ($\circ$) are extensions of $\pi$ by finitely many simple cycles using
    only transitions of $T'$. 
    
    By Theorem \ref{thm:cara}, the condition\eqref{eq:intconeX} is equivalent to 
    $\vr w - \vr v - \vr e \in \intcone{\tilde{X}}$, for a subset $\tilde{X}$ of $X$ of 
    size $\size{\tilde{X}}$ polynomial in terms of the input size.
    Therefore, in the last step, the decision procedure guesses a set $\tilde{X}$ 
    and checks whether $\vr w - \vr v - \vr e \in \intcone{\tilde{X}}$.
    The latter check is doable in $\NP$, as it is an instance of integer linear programming.
    The number of guessed items (transitions, vectors) in the whole algorithm is polynomial, and 
    therefore the final instance of integer linear programming is of polynomial size.
    In consequence, the whole decision procedure works in \NP.
\end{proof}

\section{Lower Bound}
\label{sec:lowerbound}

This section contains a proof of \pspace-hardness of the reachability problem in 
\parvass G, for any $G \leq \symg d$ of dimension $d \geq 2$.
In particular, reachability in \parvass{\symg 2} is \pspace-hard.
Later, we will indicate cases where this lower bound is tight, for instance $G = \symg d$.

\begin{lemma}\label{lem:hardness}
    For every $d \geq 2$ and $G \leq \symg d$,
    \reachprob{G} is \pspace-hard. 
\end{lemma}

\begin{proof}
    We say that a configuration $p(\vr w)$ is \emph{bounded} by $M\in\N$ if 
    $\norm{\vr w} \leq M$.
    The proof is by reduction from the \emph{bounded reachability} problem
        for \parvass 1, which is known to be $\pspace$-complete \cite{Fearnley_2015}.
        The input to this problem consists of a \parvass 1 $(\V, s, f)$
        and a number $M\in\N$ represented in binary, and we ask if  
        $\V$ has a run $s\trans{} f$ with all configurations bounded by $M$.
        For simplicity, we assume that the vector of a configuration of a \parvass 1 $q(n)$ is 
        just a nonnegative integer $n\in\N$,
        and the effect of a transition $(q, z, p)$ is just an integer $z\in\Z$.
        \Wlog~we assume that effects of transitions of a \parvass 1 are nonzero and at most $M$.
        Furthermore, \mywlog~we assume that these effects are even,
        while the number $M$ is odd.
        This guarantees that the counter value of a \parvass 1 never reaches exactly $M$.
        
        Fix an arbitrary dimension $d\geq 2$.
        Given a \parvass 1 $(\V, s, f)$ and $M$, where $\V = (Q, T)$, we
        construct a \parvass {G} $\V' = (Q', T')$ of dimension $d$
        whose states $Q' := Q \cup \setof{\pair q t}{q\in Q, t \in T}$
        extend $Q$ by additional auxiliary states, and whose transitions
        $T'$ are defined below.
        The idea is to represent a
        configuration $q(n)$ of $\V$ bounded by $M$ ($n\leq M$), by a configuration 
        \[\overline{q(n)} \ := \ q(n, \ldots, n, 2M-n)\] of $\V'$, while pairs $\pair q t$ are additional auxiliary states.
        We define
        the set $T'$ of transitions by providing a set of orbit representatives
        $\widetilde T\subseteq T'$.
        Below, we use $z\in\Npos$ to range over positive integers, and use the vectors
        \[ 
        \vr v_z = (M+z, \ldots, M+z, -M-z) \qquad\qquad
        \vr v_{-z} = (-M-z, \ldots, -M-z, M+z) 
        \]
        that have value $M+z$ (resp.~$-M-z$) on all coordinates except coordinate $d$ which has the opposite value.
        We also use $\vr v = (M, \ldots, M, -M)$.
        The set $\widetilde T$ is defined as follows:
            
        \begin{itemize}
            \item For every $t=(p, z, q) \in T$, we put into $\widetilde T$ two transitions (recall that $z\in\Npos$)
            \[
            \widetilde t = (p, \vr v_z, \pair q t) \qquad\qquad
            t^- = (\pair q t, -\vr v, q);
            \]
            \item For every $t=(p, -z, q) \in T$, we put into $\widetilde T$ two transitions 
            \[
            t^+ = (p, \vr v, \pair p t) \qquad\qquad
            \widetilde t = (\pair p t, \vr v_{-z}, q).
            \]
        \end{itemize}
        
        \noindent
        As all vectors $\vr v_z, \vr v_{-z}$ and $\vr v$ are constant except for one coordinate, 
        we notice that the size of $T' = \setof{\sigma(t)}{t\in\widetilde T, \sigma \in G}$ is at most        
        $\size{T'} \leq d \cdot \size{\widetilde T}$. 
        This is however not relevant for the reduction, as $\V'$ is represented by $\widetilde T$, and not by $T'$.
        
        We argue that there is a run $s\trans{} f$ in $\V$ whose configurations are all bounded by $M$,
        if and only if $\overline s \trans{} \overline f$ in $\V'$.
        For the `only if' direction we observe that every step $c\trans{t} c'$ in $\V$ which is bounded by $M$ 
        is simulated by two steps $\overline c\trans{} c' \trans{} \overline c'$ in $\V'$, for some intermediate
        configuration $c'$ with an auxiliary state.
        Indeed, when $t=(p, z, q)$, 
        a step $p(n) \trans{t} q(n+z)$ of $\V$ which is bounded by $M$ ($n+z\leq M$),
        is simulated in $\V'$ by the following two steps:
        \begin{align} \label{eq:2+}
        \overline{p(n)} \trans{\widetilde t} \pair q t(M+n+z, \ldots, M+n+z, M-n-z) \trans{t^-} \overline{q(n+z)}.
         \end{align}
         In the other case, when $t=(p, -z, q)$, 
        a step $p(n) \trans{t} q(n-z)$ of $\V$ which is bounded by $M$ ($n\leq M$),
        is simulated in $\V'$ by the following two steps:
        \begin{align} \label{eq:2-}
        \overline{p(n)} \trans{t^+} \pair p t(M+n, \ldots, M+n, M-n) \trans{\widetilde t} \overline{q(n-z)}.
         \end{align}
         For the `if' direction, we consider a configuration 
         $\overline{q(n)} = q(n, \ldots, n, 2M-n)$, and analyze all possible two-step runs in $\V'$ 
         starting from $\overline{q(n)}$:
         $
         \overline{q(n)} \trans{} c' \trans{} c''.
         $
         Recall that $n < M$.
         By definition of $T'$, we deduce that 
         the only possible such two-step runs are of the form \eqref{eq:2+} or \eqref{eq:2-}.
         Indeed, the first transition may be only $\widetilde t$ for some $t=(q, z, p)\in T$,
         or $t^+$ for some $t=(q,-z,p)\in T$, and then the next transition is uniquely determined
         by the auxiliary state.
         In particular, we observe that only transitions from $\widetilde T$ can be used, i.e.,
         all other transitions $T' \setminus \widetilde T$ are useless if one starts from a configuration
         of the form $\overline{p(n)}$.
         In consequence, if $\overline s\trans{} \overline f$ in $\V'$, we have also
         a run $s\trans{} f$ in $\V$ all of whose configurations are bounded by $M$, as required.
         Correctness of the reduction follows.
        \end{proof}

\begin{remark}
In particular, we have \pspace-hardness for \parvass{\symg 2}, which is 
an improvement over \pspace-hardness for \parvass 2 \cite{Blondin15}.
\end{remark}


\begin{remark}\label{rem:bounded_hard}
Notice that all reachable configurations of the \parvass{\symg d} constructed in 
the proof of Lemma \ref{lem:hardness} are bounded by $2M$.
We thus deduce that the lower bound applies to \parvass G $(\V, s)$ which are 
\emph{polynomially bounded},
namely bounded by some $B$ polynomial in $\norm{\V, s}$.
\end{remark}

\section{Transitive and Fair Groups}
\label{sec:transitive}

Within this section we fix an arbitrary \emph{transitive} group $G \leq \symg d$,
i.e., we assume that for every $i, j \in \setto d$ there exists $\sigma \in G$
such that $\sigma(i) = j$.
The results of this section will be applied later to specific transitive groups $G$.

\para{Pumpable \vass}

Consider a \parvass G $(\V, s, t)$, where $s=q(\vr w)$ and $t=q'(\vr w')$.
We say that 
$(\V, s)$ is \emph{forward pumpable} if $\V$ has a run $q(\vr w) \trans{} q(\vr w+\vr e)$ 
for some $\vr e \geq \const 1$.
Symmetrically, $(\V, t)$ is \emph{backward pumpable} if 
$(\rev \V, t)$ is forward pumpable where $\rev \V$, the \emph{reverse} of $\V$,
is obtained by replacing each transition $(q, \vr v, q')$ of $\V$ by its reverse
$(q', -\vr v, q)$.
Equivalently,  $\V$ has a run $q'(\vr w'+\vr e') \trans{} q'(\vr w')$ 
for some $\vr e' \geq \const 1$.
Finally, $(\V, s, t)$ is \emph{pumpable} if $(\V, s)$ is forward pumpable, and $(\V, t)$ is backward pumpable.
We prove that reachability reduces to $\Z$-reachability,
when the group $G$ is transitive and \vass are pumpable.
\begin{lemma} \label{lem:transitive}
    If $G$ is transitive then
    every pumpable \parvass G $(\V, s, t)$ admitting a \Zrun from $s$ to $t$, admits
    a run from $s$ to $t$.
\end{lemma}
\begin{proof}
    Let $G\leq \symg d$.
    We start by proving the following equality for all $i, j \in \setto d$:
    \begin{equation}\label{eq:uniform}
        \size{\setof{\sigma \in G}{\sigma(i) = j}} \ = \ \frac{\size{G}}{d}.
    \end{equation}
    When $i = j$, the equality follows immediately from the Orbit-Stabilizer Theorem.
    Otherwise, suppose $i \neq j$ and take any permutation $\sigma_{ij} \in G$ 
    such that $\sigma_{ij}(i) = j$ (it exists as $G$ is transitive).
    By post-composing with $\sigma_{ij}$ we get
    a bijection $\sigma \mapsto \sigma_{ij} \circ \sigma$
    between permutations $\sigma\in G$ that satisfy $\sigma(i) = i$, and permutations $\sigma\in G$ 
    that satisfy
    $\sigma(i)=j$.
    In consequence, we get the equality \eqref{eq:uniform}:
    \begin{equation*}
        \size{\setof{\sigma \in G}{\sigma(i) = j}} \ = \ \size{\setof{\sigma \in G}{\sigma(i) = i}} \ =  \ 
        \frac{\size{G}}{d}.
    \end{equation*}

    Consider a \parvass G $\V$ and two configurations $s = q(\vr v)$ and $t=q'(\vr v')$, and suppose
    $\V$ has a \Zrun  $\gamma$ from $q(\vr v)$ to $q'(\vr v')$, 
    and two runs 
    \[
    \alpha \colon q(\vr v) \trans{} q(\vr v + \vr e) \qquad\qquad 
    \alpha' \colon q'(\vr v'+\vr e') \trans{} q'(\vr v'),
    \] 
    for some $\vr e, \vr e'\geq \const 1$.
    We argue that  $\V$ has a run $q(\vr v) \trans{} q'(\vr v)$.
    \Wlog we can assume that $\alpha$ is \emph{increasing}, i.e., every coordinate of $\vr v + \vr e$ is greater than
    the greatest coordinate of $\vr v$ (otherwise we replace 
    $\alpha$ by its $n$-fold concatenation $\alpha^n \colon q(\vr v) \trans{} q(\vr v + n\vr e)$, for sufficiently large $n \in \N$).
    Symmetrically, we can assume that $\alpha'$ is \emph{decreasing}, i.e.,
    every coordinate of $\vr v' + \vr e'$ is greater than the greatest coordinate of $\vr v'$.

    Let $G = \set{\sigma_1, \ldots, \sigma_m}$, where $\sigma_1 = \id d$ is the identity
    permutation on $\setto d$,
    and let
    \begin{equation*}
    \Delta \ := \ \sum_{i = 1}^{m} \sigma_i(\vr e),
    \qquad \qquad
    \Delta' \ := \ \sum_{i =1}^m \sigma_i(\vr e').
    \end{equation*}
    As $\alpha$ is increasing and $\sigma_1(\alpha) = \alpha$, the concatenation of runs
    \[
    \widetilde \alpha  \ = \ \sigma_1(\alpha) ; \sigma_2(\alpha) ; \ldots ; \sigma_{m}(\alpha)
    \] 
    is forcedly a run from $q(\vr v)$ to $q(\vr v + \Delta)$.
    Symmetrically, as $\alpha'$ is decreasing  and $\sigma_1(\alpha') = \alpha'$,
    the concatenation of runs
    $
    \widetilde \alpha'  \ = \ \sigma_{m}(\alpha') ; \ldots ; \sigma_2(\alpha') ; \sigma_{1}(\alpha')
    $ 
    is forcedly a run from $q'(\vr v' + \Delta')$ to $q'(\vr v')$.
    We argue that $\Delta(j) = \Delta(k)$ for every $j,k \in \setto d$, i.e., the effect of 
    $\widetilde\alpha$ is a constant vector.
    Indeed, by point-wise expansion of the equality
    \[
    \Delta \  = \ 
    \sum_{i=1}^m \sigma_i(\vr e)  \ = \ \sum_{i=1}^m \vr e \circ \sigma_i,
    \]
    and using the equality \eqref{eq:uniform} to deduce that for every $\ell$, the value of $\ell$th 
    coordinate
    $\vr e(\ell)$ of $\vr e$ appears the same number of times on both sides of the equality below, we compute:
    \[
    \Delta(j) \ = \ \sum_{\ell = 1}^d \ \sum_{\substack{i\,:\,\sigma_i(\ell) = j}} \vr e(\ell)
    \ \ \stackrel{\eqref{eq:uniform}}{=} \ \ 
    \ \ \sum_{\ell = 1}^d \ \sum_{\substack{i\,:\,\sigma_i(\ell) = k}} \vr e(\ell)
    \ = \ \Delta(k).
    \]
    Therefore, $\Delta = \const n$ for some $n\in\Npos$.
    Analogously we prove that $\Delta' = \const{n'}$ for some $n'\in\Npos$.
    Let $\widetilde \Delta = \const{(n\cdot n')}$.
    Since $\widetilde \Delta = n'\cdot \Delta = n \cdot \Delta'$,
    we may construct two  runs that pump forward and backward by the same vector $\widetilde \Delta$:
    \[
    \beta = \widetilde \alpha^{(n')} \colon q(\vr v) \trans{} q(\vr v + \widetilde\Delta) \qquad\qquad
    \beta' = (\widetilde \alpha')^n \colon q'(\vr v' + \widetilde\Delta) \trans{} q'(\vr v').
    \]
    Thus, for sufficiently large $m\in\Npos$, the $m$-fold iterations of $\beta$ and $\beta'$ 
    are enough to lift the \Zrun $\gamma$ into a run, i.e., to ensure that
    $\beta^m ; \gamma ; (\beta')^m$ is a run from $q(\vr v)$ to $q'(\vr v')$.    
\end{proof}

\para{\Fair groups}

Given $d\in\Npos$ and a transitive permutation group $G\leq \symg d$,
we say that $G$ is \emph{\fair} if 
there is a polynomial $P$
such that
for every $R\in\N$ and every \parvass G $(\V,s)$ with $S := \norm{\V, s}$ and $N := \norm{\V}$,
the implication $(1) \implies (2)$ holds, where

\begin{itemize}
\item[(1)]
$\V$ has a run 
$s \trans{} p(\vr w)$ for some $p\in Q$ and $\vr w$ of norm
$\norm{\vr w} > P(N)\cdot (S+R)$,
\item[(2)]
$\V$ has a  run $s \trans{} p(\vr w)$ for some $p\in Q$ and 
$\vr w \geq \const{S + R+1}$.
\end{itemize}

\noindent
In words:
whenever $(\V,s)$ has a run whose target vector exceeds $P(N)\cdot (S+R)$ on \emph{some} coordinate,
$(\V,s)$ is guaranteed to have a run whose target vector is greater than $S+R$ on \emph{all} coordinates;
and the property holds for any chosen $R\in\N$. 

\begin{lemma}\label{lem:pspace}
    If $G$ is fair
    then \reachprob G is in \pspace.
\end{lemma}

\begin{proof}    
We start with some preparatory definitions.
    A \vass is called \emph{strongly connected} if its state graph is so.
A \vass $\V = (Q, T)$ is called \emph{sequential},
if it can be partitioned into a number of strongly connected 
\vass $\V_1=(Q_1, T_1), \ldots, \V_k=(Q_k, T_k)$ with pairwise disjoint state spaces, 
called \emph{components} of $\V$,
and $k-1$ transitions
$u_i=(q_i, \vr v_i, q'_i)$, for $i\in\setto{k-1}$, where $q_i\in Q_i$ and $q'_i \in Q_{i+1}$,
called \emph{bridges} of $\V$
(see Fig.~\ref{fig:kcompvass}).
For every non-sequential \vass $(\V, s, t)$ there is a finite set  $W=\{\V^1, \ldots, \V^\ell\}$ such
that each $\V^j$ is a sequential \vass, and
 $s\trans{} t$ in $\V$ if and only if $s\trans{} t$ in $\V^j$ for some $j\in\setto\ell$. 
The set $W$, called the \emph{sequential decomposition} of $\V$ in the sequel, 
may be computed based on the decomposition of the state graph of $\V$ into strongly connected
components.

\begin{figure}
\begin{center}
\scalebox{0.75}{
\begin{tikzpicture}[shorten >=1pt,node distance=3cm,on grid,>={Stealth[round]},
    every state/.style={draw=blue!50,thick,fill=blue!20},scale=0.25]

  \node[state]          (q_0)                      {$\V_1$};
  \node[state]          (q_1) [right=of q_0] {$\V_2$};
  \node[]                  (q_ldots) [right=of q_1] {\ \ \ \ \ \ ... \ \ \ \ \ \ };
  \node[state]          (q_k) [right=of q_ldots] {$\V_k$};

  \path[->] 
            (q_0) edge              node [above]  {$u_1$} (q_1)
            (q_1) edge              node [above]  {$u_2$} (q_ldots)
            (q_ldots) edge              node [above]  {$u_{k-1}$} (q_k);
\end{tikzpicture}
}
\end{center}
\caption{A sequential \vass.}
\label{fig:kcompvass}
\end{figure}
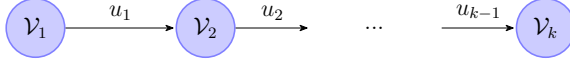

\smallskip

    Assume an input  \parvass G $(\V, s, t)$, where
    $\V = (Q, T)$, $s = q(\vr v)$ and $t = p(\vr w)$, represented by 
    orbit representatives $\widetilde T \subseteq T$.
    The \pspace decision procedure to check if $s\trans{} t$, to be described below, proceeds in several steps.

\smallskip

    As the first step, the decision procedure decomposes the state graph of $\V$ 
    into strongly connected components, and
    nondeterministically picks up one sequential \vass from the sequential decomposition of $(\V,s,t)$.
    From now on we may thus assume, \mywlog, that $\V$ is sequential with $k$ components
    $\V_1, \ldots, \V_k$, the configuration $s$ belongs to the first
    component, and $t$ to the last one.
    The number $k$ of components is bounded by the number $\size{Q}$ of the states of $\V$.
    Let $R := \size Q \cdot \norm{\V}$.
    Thus, for every two states $p, p'$ from the same component of $\V$,
    there is a path from $p$ to $p'$ that decreases each counter by at most $R$.
    
    \smallskip
    
    Recall that, for $M\in\Npos$, we say that a configuration $p(\vr w)$ is 
    \emph{bounded} by $M$ if $\norm{\vr w} \leq M$. 
    We also say that the first component $(\V_1, s)$  is bounded by $M$ if every configuration $c$ 
    reachable in this component from $s$ (i.e, $c$ satisfying $s\trans{} c$) is bounded by $M$.
    Symmetrically, we say that
    the last component $(\V_k, t)$ is reverse bounded by $M$ if every configuration $c$ reverse
    reachable from $t$ (i.e., $c$ satisfying $c\trans{} t$) in this component is bounded by $M$.
   
\smallskip
 
    As the second step, the decision procedure iteratively removes either the first or the last component from
    $(\V,s,t)$ until the following two conditions are met, for $N := \norm{\V}$ and $S = \norm{\V, s, t}$:
    
    \begin{itemize}
    \item   
    the first component is \emph{not bounded} by $P(N)\cdot (S+R)$, and
    \item
    the last component is \emph{not reverse bounded} by $P(N)\cdot (S+R)$,
    \end{itemize}

    \noindent
    where $P$ is the polynomial given by the definition of a \fair{} group.
    In particular, the first condition implies (1) for $(\V_1, s)$, and the second one implies (1) for $(\rev \V_k, t)$.
    Suppose  
    the first component is \emph{bounded}
    by $P(N)\cdot (S+R)$ (the symmetric case, when the last component is reverse bounded by $P(N)\cdot (S+R)$, is 
    treated in the same way).
    One iteration of the removal procedure proceeds as follows, depending on whether $k=1$ or $k>1$:
    
    \para{Case I: $k=1$}

    If this component is the only one, the procedure
    checks in \pspace if $s\trans{} t$,  and terminates.
    The check is done by a standard nondeterministic walk through the graph of configurations reachable from $s$,
    which starts at the source configuration $s$, and in every iteration moves to a nondeterministically chosen 
    successor of the current configuration.
    It terminates when $t$ is reached (and answers positively), or when the counter of the visited configurations so far
    exceeds \[ \size Q \cdot \big(P(N)\cdot (S+R)+1\big)^d,\]
    the total number of configurations bounded by $P(N)\cdot (S+R)$.
    The memory size, namely the bitsize of a single configuration, and of the value of the counter, 
    are both polynomially bounded, and hence the
    whole procedure works in (nondeterministic) \pspace.
     
    \para{Case II: $k>1$}

    Otherwise, 
    let $u_1 = (p, \vr u, p')$ be the first bridge.
    We observe that
    if $\V$ has a run $s \trans{} t$, then its last configuration $p(\vr w)$ in the first component
    is necessarily bounded by $P(N)\cdot (S+R)$.
    In consequence, since the bridge transition may increase the norm by at most $N$, the first configuration
    $p'(\vr w')$ of the run in the second component, where $\vr w'=\vr w+\vr u$, 
    is forcedly bounded by $P(N)\cdot (S+R) + N$.
    Relying on this observation, the decision procedure nondeterministically chooses a vector $\vr w$
    with $\norm{\vr w}\leq P(N)\cdot (S+R)$, and checks in \pspace if $q(\vr v) \trans{} p(\vr w)$
    in the first component, similarly as above.
    Then the procedure removes the first component from $\V$, and takes 
    $s' = p'(\vr w') = p'(\vr w+\vr u)$ as the source configuration, thus obtaining a new \parvass G 
    $(\V', s', t)$ with one less component than $(\V, s, t)$.
    
\smallskip
    
    This removal is iteratively repeated until either the procedure terminates, or produces
    a $(\V', s', t)$ with the first component \emph{not bounded} by $P(N')\cdot (S'+R)$,
    and the last component \emph{not reverse bounded} by $P(N')\cdot (S'+R)$, 
    where $N' := \norm{\V'}$ and $S' := \norm{\V', s', t'}$.
    The bound on the norm $N'$ of the \vass cannot increase, and so $N' \leq N$.
    However, in each iteration the bound on the norm $S'$ of the new source vector $\vr w'$ increases, 
    compared to the bound $S$ on the norm
    of the previous source vector $\vr v$, by a multiplicative factor $\OO(P(N) \cdot \size{Q})$, namely 
    \[
    S' \ \leq \ P(N)\cdot (S+R) + N \ \leq \ (P(N) \cdot (1+\size{Q}) +1) \cdot S,
    \] 
    and therefore the bitsize of the source vector increases by (bounding roughly) at most $\OO(\log (P(N) \cdot \size{Q})) \cdot d$.
    As the number of iterations is bounded by $\size Q$, the final bitsize is polynomial in the input size,
    and therefore all the iterations are doable in \pspace.
    Note that due to iterative increase of $S'$, the first (resp.~last) component not bounded
    (resp.~not reverse bounded) by $P(N')\cdot (S'+R)$  in some iteration may become bounded (resp.~reverse bounded) in
    subsequent iterations.
    
\smallskip

    Finally, the procedure arrives at a sequential \parvass G $(\V, s, t)$,
    $s=q(\vr v)$, $t = p(\vr w)$,
    whose first component is \emph{not bounded} by $P(N)\cdot (S+R)$, and
    whose last component is \emph{not reverse bounded} by $P(N)\cdot (S+R)$.
    As 
    $G$ is \fair{},
    there is a configuration $c$ reachable from $s$ in the first component,
    and a configuration $c'$ reverse reachable from $t$ in the last component
    (i.e., $t$ is reachable from $c'$), both with vectors $\geq \const{S+R+1}$: 
    \[
    q(\vr v) \trans{} c \qquad\qquad
    c' \trans{} p(\vr w).
    \]
    As $R$ is large enough so that every two states in the same component
    are connected by a path that decreases each counter by at most $R$,
    we have some runs
    \[
    c\trans{} q(\overline {\vr v}) \qquad \qquad
    p(\overline {\vr w}) \trans{} c'
    \]
    in the first and the last component, respectively, with $\overline{\vr v}, \overline {\vr w} \geq \const{S+1}$.
    As $\norm {s}, \norm {t} \leq S$, this means that $\V'$ is pumpable. 
    Since $G$ is transitive, Lemma \ref{lem:transitive} applies.    
    In the last step, the decision procedure checks if there is a \Zrun from $s$ to $t$,
    which is doable in \NP due to Lemma \ref{lem:zrun}.
\end{proof}

\section{Symmetric Group} \label{sec:symg}

In this section we prove \pspace-completeness
of the reachability problem in case of
the symmetric groups $\symg d$,  $d\geq 2$.
Our \pspace upper bound works not only for arbitrary fixed dimension $d\geq 2$, but even
in the uniform setting when $d$ is part of the input. 
When $d=1$, the problem is known to be \NP-complete \cite{HKOW09}.
As $\altg d \leq \symg d$, the \pspace upper bound
also follows from Theorem \ref{thm:An} stated in Section \ref{sec:other}.
However, the group $S_d$ is central to this paper, and we therefore give a separate proof.

The upper bound is shown using Lemma \ref{lem:pspace}, which requires proving that $\symg d$
satisfies its assumptions.
\begin{lemma}\label{lem:sd_bal}
    $\symg d$ is \fair, for every $d\in\Npos$.
\end{lemma}
\begin{proof}
    $\symg d$ is obviously transitive,
    hence we concentrate on proving the implication $(1) \implies (2)$
    from the definition a fair group.

    Consider an \svass{$\symg d$} $\V = (Q, T)$ and a configuration $s = q(\vr v)$.
    Let $S := \norm{\V, s}$,
    let $R\in\N$ be any constant, and
    $B := 3S + 2R$.
    We demonstrate that $\symg d$ is \fair, and that this property
    is witnessed by the constant polynomial $P(x)=3d$. 
    Towards this we prove a slightly stronger fact:
    we assume that $\V$ has a run $\pi \colon q(\vr v) \trans{} p(\vr w)$ with $\norm{\vr w} > B\cdot d$,
    and construct a run $\pi' : q(\vr v) \trans{} p(\vr w')$ whose target vector satisfies 
    $\vr w' \geq \const{S+R+1}$. 

    \Wlog~assume that $\vr w(d) > B\cdot d$.
    For $i = 0,1, \ldots, d-1$ we inductively construct runs $\pi_i \colon q(\vr v) \to p(\vr w_i)$ satisfying
    \begin{align} \label{eq:wi}
    \vr w_i(d) > B\cdot (d-i) \qquad \qquad 
    \vr w_i(j) > S + R \ \ (\text{for all } j \leq i).
    \end{align}
    We start with $\pi_0 := \pi$.
    Assuming $\pi_{i-1}$ has already been constructed, we construct $\pi_{i}$ as follows.
    If $\vr w_{i-1}(i) > S + R$ then we take $\pi_{i} := \pi_{i-1}$.
    Otherwise, we split $\pi_{i-1}$ into 
    $\pi_{i-1} \ = \ \alpha ; \beta,$
    \[
    q(\vr v) \trans{\alpha} r(\vr u)
    \qquad\qquad
    r(\vr u) \trans{\beta} p(\vr w_{i-1}),
    \]
    where 
    $\beta$ is the longest suffix of the run
    in which the value of coordinate $d$ does not drop below $B\cdot(d-i)$.

    \begin{figure}[th]
    \begin{center}
    \begin{tikzpicture}
        \begin{axis}[
            axis line style = {draw=none},
            xlabel = {},
            ylabel = {},
            samples = 100,
            domain = 0:6,
            enlargelimits=true,
            xtick = \empty,
            ytick = \empty,
            y post scale=0.6
        ]
            \addplot[blue, dashed] {0.3*x^3 - 2.5*x^2 + 5.5*x + 0.5};

            \addplot[only marks, mark=*, mark size = 1pt] coordinates {(4.7,2.22)} node[above right] {};
            \addplot[only marks, mark=*, mark size = 1pt] coordinates {(0,0.5)} node[above left] {};
            \addplot[only marks, mark=*, mark size = 1pt] coordinates {(6,8.2)} node[above right] {};

            \node at (axis cs:2.65, 2.9) [anchor=south west] {\scriptsize $\vr u(d) \geq B\cdot(d-i)$};
            \node at (axis cs:0.2, 0.2) [anchor=south west] {\scriptsize $\vr v(d)$};
            \node at (axis cs:2.3, 6.9) [anchor=south west] {\scriptsize $\vr w_{i-1}(d) > B\cdot(d-i+1)$};
        
            \draw[thin] (axis cs:0,0) -- (axis cs:4.5,0) node[midway,above] {$\alpha$};
            \draw[thin] (axis cs:4.7,0) -- (axis cs:6,0) node[midway,above] {$\beta$};

            \draw[thin] (axis cs:0,-0.3) -- (axis cs:0,0.3);
            \draw[thin] (axis cs:4.5,-0.3) -- (axis cs:4.5,0.3);
            
            \draw[thin] (axis cs:4.7,-0.3) -- (axis cs:4.7,0.3);
            \draw[thin] (axis cs:6,-0.3) -- (axis cs:6,0.3);
        \end{axis}
    \end{tikzpicture}
\end{center}
\label{fig:sd_step}
\end{figure}

\noindent
    Therefore, $\vr u(d) < B\cdot(d-i)+S$, as otherwise $\beta$ would be longer.
    Let $\Delta = \vr w_{i-1} - \vr u$ be the difference between the target and source of $\beta$.
    We observe that the value of coordinate $d$
    is increased by $\beta$ by at least
    \[
     \Delta(d) \ > \ B\cdot(d-i+1) - B\cdot(d-i)-S = 2S+2R,
    \]
    while the value of coordinate $i$ is increased by at most 
    $
    \Delta(i) \ \leq \ \vr w_{i-1}(i) \ \leq \ S+R.
    $
    Therefore, $\Delta(d) - \Delta(i) > S+R$.
    In consequence, $\beta$ contains some $k$ steps
    \[
    q_1(\vr u_1) \trans{t_1} q'_1(\vr u'_1) \quad \ldots \quad
    q_k(\vr u_k) \trans{t_k} q'_k(\vr u'_k)
    \]
    induced by transitions $t_1 = (q_1, \vr e_1, q'_1), \ldots, t_k = (q_k, \vr e_k, q'_k)$
    with effects 
    \[
    \vr e_1  \ = \ \vr u'_1 - \vr u_1 \qquad \ldots \qquad \vr e_{k} \ = \ \vr u'_k - \vr u_k,
    \]
    such that $\vr e_\ell(d) > \vr e_\ell(i)$ for $\ell \in \setto {k}$, that is, each of the steps increases 
    the difference between 
    the values of coordinate $d$ and $i$, and moreover
    \begin{align} \label{eq:NR1}
    \vr e_1(d) + \ldots + \vr e_k(d) \ > \ \vr e_1(i) + \ldots + \vr e_k(i) + S+R,
    \end{align}
    that is, the steps jointly increase the difference between the values of coordinate $d$ and $i$ 
    by more than $S+R$.
    We may choose a \emph{minimal} set of transitions $t_1, \ldots, t_k$ satisfying \eqref{eq:NR1}.
    Hence, knowing that one step can change the difference between the values of coordinate $d$ and $i$
    by at most $2S$, we may safely assume
    \begin{align} \label{eq:3NR}
    \vr e_1(d) + \ldots +\vr e_k(d) \ \leq \ \vr e_1(i) + \ldots + \vr e_k(i) + 3S+R,
    \end{align}
    that is, the steps jointly increase the difference between the values of coordinate $d$ and $i$ 
    by at most $3S+R$.
    Let $\sigma = {\tiny \big( \begin{array}{c} i \ \ d \\ d \  \ i \end{array} \big)} \in S_d$  
    be the permutation that swaps coordinates $d$ and $i$ and preserves the others.
    We define a run $\pi_{i} \colon q(\vr v) \runarrow p(\vr w_{i})$
    by replacing each transition $t_\ell$ in $\pi_{i-1}$ with the transition $\sigma(t_\ell)$.
    Due to the inequality \eqref{eq:3NR},
    the coordinate $d$ never drops below 
    \[
    B\cdot(d-i) - (3S+R) \ > \ 0
    \] 
    in $\pi_{i}$, 
    and $\vr w_{i}(d) > B\cdot(d-i+1) - (3S+R) > B\cdot(d-i)$.
    Furthermore, due to the inequality \eqref{eq:NR1} we get $\vr w_{i}(i) > S + R$, 
    because the coordinate $i$ was increased by more than $S+R$.
    Other coordinates are not affected, i.e., their values are the same in $\pi_{i-1}$ and $\pi_{i}$.
    Therefore, the run $\pi_{i}$ satisfies the condition \eqref{eq:wi}, as required.
    
    Finally, we arrive at a run $\pi' = \pi_{d-1} \colon q(\vr v) \runarrow p(\vr w_{d-1})$
    with $\vr w_{d-1}(i) \geq S+R+1$ for every $i \in \setto d$, and this completes the proof.
\end{proof}

\begin{theorem}\label{thm:Sn}
    \reachprob{\symg d} is \pspace-complete, for every $d\geq 2$.
\end{theorem}

\begin{proof}
    Membership in \pspace follows immediately by Lemmas \ref{lem:pspace} and \ref{lem:sd_bal}, and
    \pspace-hardness follows by Lemma \ref{lem:hardness}.
\end{proof}


\section{Combining Groups} 
\label{sec:comb}

In this section we investigate groups arising as combinations of smaller 
permutation groups.
The way of combining two permutation groups $G\leq \symg g$ and $H\leq\symg h$
will be their \emph{wreath product} $\wrprod G H \leq \symg {g h}$.
The set $\setto{g h} = \{1, \ldots, g h\}$ we conveniently split into $h$ \emph{blocks}:
\[
\setfromto{1}{g} \quad
\setfromto{g+1}{2g} \quad \ldots \quad
\setfromto {g(h-1)+1}{gh},
\]
and write $\pair i j$ to denote $g(j-1) + i\in\setto{gh}$, for $i\in\setto g$ and $j\in\setto h$.
Intuitively, $\pair i j$ stands for the $i$th element of the $j$th block.
Given a $h$-tuple $(\sigma_1, \ldots, \sigma_h)$ of permutations from $G$
and one permutation $\delta$ from $H$, 
we define the permutation
$\wrtuple{\sigma_1, \ldots, \sigma_h} \delta \in \symg {g\cdot h}$
that applies, for each $j$,
the permutation $\sigma_j$ on the block $\setfromto{g(j-1)+1}{gj}$ independently,
plus permutes the blocks according to $\delta$,
namely
\[
\wrtuple{\sigma_1, \ldots, \sigma_h} \delta \ : \ \pair i j \ \ \longmapsto \ \ \pair {\sigma_j(i)} {\delta(j)}.
\]
The \emph{wreath product}%
\footnote{
The operation is definable more concisely, as
semidirect product of $h$-fold direct product of $G$ with $H$.
}
 of $G$ and $H$ consists of all so constructible permutations
of $\setto{gh}$:
\[
\wrprod G H \ := \ 
\setof{\wrtuple{\sigma_{1}, \ldots, \sigma_{h}} \delta}{\sigma_1, \ldots, \sigma_h \in G, \ \delta\in H}.
\]

In the sequel we are mostly interested in combining trivial permutation groups $\tg n$
with symmetric groups $\symg d$, and relating the complexity of the reachability
problem for the combined group to the complexity of the problem for
component groups.
As we show in Theorems \ref{thm:large} and \ref{thm:small} below,
it turns out that the complexity significantly depends on the order of composing the groups.
In the case of the group $G = \wrprod {\tg n} {\symg d}$ which allows for arbitrary permutations
of $d$ blocks (of size $n$ each) but disallows permutations inside blocks, the complexity of \reachprob G 
is at least as high as the complexity of
\reachprob {\big((n-1)d\big)}
(Theorem \ref{thm:large}).
On the other hand, in the dual case of the group $G = \wrprod {\symg d} {\tg n}$, 
which disallows permutations of $n$ blocks (of size $d$ each)
but allows arbitrary simultaneous independent permutations inside blocks,
we prove (in Theorem \ref{thm:small} below) that the  complexity of \reachprob G 
drops to the complexity of
\reachprob {n}.
Summing up, in the former case 
the impact of symmetric group $\symg d$ is essentially nullified as
both degrees $n$ and $d$ contribute significantly to the complexity
of the combined group, while in the latter case 
the impact of symmetric group $\symg d$ is preserved,
as the degree $d$ of the symmetric group
is irrelevant for complexity.
 We notice that the difference of complexities is in agreement with the fact
 that up to the isomorphism $\pair i j \mapsto \pair j i$  of permutation groups, 
$\wrprod {\tg n} {\symg d}$ is a subgroup of 
$\wrprod {\symg d} {\tg n}$.

\para{Trivial symmetry inside blocks}

The first result we formulate slightly more generally, 
with an arbitrary group $G\leq\symg d$ in
place of the symmetric group $\symg d$.
The reachability problem
for wreath product of the trivial permutation group $\tg n$ 
with any  $G\leq \symg d$, 
regardless of the choice of $G$ (which seems to be quite surprising),
is at least as hard as the reachability problem
for \vass of dimension $(n-1)d$:

\begin{theorem}\label{thm:large}
    \reachprob{\big((n-1)d\big)}
    reduces in polynomial time
    to \reachprob{(\tg n \wr G)},
    for every $n\geq 2$, $d\geq 1$ and $G\leq\symg d$.
\end{theorem}
\begin{proof}
    Given a  \vass{}  $\V=(Q,T) $ of dimension $(n-1)d$, we define
    a \svass{$(\tg n \wr G)$} $\V'=(Q',T')$ simulating $\V$,
    whose states extend $Q$ by a number of auxiliary states: 
    \[Q' \ = \ Q \ \cup \ \setof{\tuple{q,i,t}}{q \in Q, i \in \setto{d-1}, t \in T}.\]
    A configuration $q(\vr w)\in Q\times \N^{(n-1)d}$ of $\V$, will be simulated
    by a configuration $\overline{q(\vr w)} = q(\overline{\vr w}) \in Q \times \N^{n d}$ of $\V'$ whose vector 
    $\overline{\vr w}$ extends $\vr w$, on the last $n$th
    coordinate of every block, by
    \begin{align} \label{eq:over}
    \overline{\vr w}(n, j) = j-1 \qquad (\text{for } j \in \setto d)
    \end{align}
    (we write $\overline{\vr w}(n, j)$ instead of $\overline{\vr w}(\pair n j)$, for readability).
    We define transitions $T'$ of $\V'$ by its orbit representatives $\widetilde T \subseteq T'$. 
    For every vector $\vr v \in \Z^{(n-1)d}$ we define $d$ vectors $\vr v_1, \ldots, \vr v_d \in \Z^{nd}$ as follows:
    \[
    \vr v_k(i, j) \ = \ \begin{cases}
    -d+1 & i = n, j=k \\
    1 & i = n, j \neq k \\
    \vr v(i, j) & i < n, j = k \\
    0 & i<n, j\neq k.
    \end{cases}
    \]
    Thus, the value $\vr v_k(n,k)$ on the last coordinate of $k$th block equals $-d+1$, 
    while on the last coordinate of other blocks equals  $1$.
    Furthermore, $\vr v_k$ coincides with $\vr v$ on the remaining $n-1$ coordinates of $k$th block, 
    while it is 0 on the remaining $n-1$ coordinates of other blocks.
    For every  transition $t = (p, \vr v, q) \in T$,
    we put to the set $\widetilde T$ the following $d$ transitions:
    \begin{align*}
    t_d  \ = \ & (p, \vr v_d, \tuple{q,d-1,t}), \\
    t_{d-1} \ = \ & (\tuple{q, d-1, t}, \vr v_{d-1}, \tuple{q,d-2,t}), \\
    & \ldots \\
    t_{2} \ = \ & (\tuple{q, 2, t}, \vr v_2, \tuple{q,1,t}), \\
    t_{1} \ = \ & (\tuple{q, 1, t}, \vr v_1, q).    
    \end{align*}
    We need to prove that there is a run $p({\vr w}) \trans{} q({\vr v})$ in $\V$ if and only if
    there is a run $\overline {p(\vr w)} \trans{} \overline {q(\vr v)}$ in $\V'$.
    The `only if' direction follows due to observation 
    that every step $p({\vr w}) \trans{t} q({\vr v})$ in $\V$
    is simulated
    by a sequence of transitions $\alpha_t = t_d, t_{d-1}, \ldots, t_1$, namely
    \begin{align} \label{eq:segm}
    \overline {p(\vr w)} \trans{\alpha_t} \overline {q(\vr v)}
    \end{align} 
    in $\V'$,
    as every transition $t_k$ updates coordinates $1\ldots n-1$ of block $k$, and the effect
    of $\alpha_t$ on last $n$th coordinate of every block is 0.
    
    For the `if' direction, we observe that due to the auxiliary states in $Q'$, every
    run $\overline {p(\vr w)} \trans{} \overline {q(\vr v)}$ in $\V'$, starting and ending in a state from $Q$,
    necessarily splits into \emph{segments} of length $d$, where
    each segment is  of the form 
    \begin{align}\label{eq:toprove}
    \sigma_d(t_d), \sigma_{d-1}(t_{d-1}), \ldots, \sigma_1(t_1),
    \end{align}
    for some $\sigma_1, \ldots, \sigma_d \in \tg n \wr S_d$ and $t = (p, \vr v, q) \in T$.
    In particular, the segment starts in $p\in Q$ and ends in $q\in Q$.
    It is thus enough to argue that 
    whenever $p(\overline {\vr w}) \trans{} q(\vr v')$ in $\V'$
    along a single segment \eqref{eq:toprove}, there is a vector $\vr v$ such that
    $\vr v' = \overline {\vr v}$ and
    $p({\vr w}) \trans{} q({\vr v})$ in $\V$
    ($*$). 

    A single segment updates coordinates $1\ldots n-1$ in all the $d$ blocks, in the order determined by 
    $\sigma_1, \ldots, \sigma_d$.
    The crucial observation is that this order is necessarily $d, d-1, \ldots, 1$.
    Indeed,
    by the definition of the vector $\overline{\vr w}$ in \eqref{eq:over}, the first update is in the $d$th block
    as the last $n$th coordinate is $\geq d-1$ only in this block.
    As a result of this update, the $(d-1)$th block is the only one with value $\geq d-1$ on the last $n$-th coordinate,
    and therefore the second update necessarily happens in this block.
    And so on, until the last update in the first block.
    In consequence, $\sigma_i(t_i) = t_i$ for all $i \in \setto d$.
    Therefore, the segment \eqref{eq:toprove} is necessarily of the form \eqref{eq:segm}, and the implication ($*$) is shown.
    In particular, we have shown that only transitions from $\widetilde T$ can be used in $\V'$ when starting
    from a configuration $\overline{p(\vr w)}$.
\end{proof}

\para{Trivial symmetry between  blocks}

Now we investigate the other case of  wreath product of 
the symmetric  group of degree $d$ with
the trivial permutation group of degree $n$, and show the 
complexity of the reachability problem in this case
to be essentially the same as of the reachability problem
for \vass of dimension $n$.

We start with a definition and a lemma, both stated for any $G\leq \symg n$ in place of $\tg n$.
A configuration $q(\vr w)$ of a \parvass{\big(\symg d \wr G\big)} $\V$
is \emph{$B$-\balrun} if  
\begin{equation}\label{eq:balconf}
\absv{\vr w(i, j) - \vr w(i', j)} \leq B
\end{equation}
for every $i,i' \in \setto d$ and $j \in \setto n$.
In words: in every block, the maximal difference of values is at most $B$.
A run $s\trans{} t$ of $\V$ is $B$-\balrun if  all its configurations are so.
Finally, $(\V,s,t)$ is \emph{$B$-\balrunable} if the existence of a run $s\trans{} t$ implies
the existence of  a run $s\trans{} t$ which is $B$-\balrun.
\begin{lemma}\label{lem:balancing}
    There is 
    a polynomial $P$ such that for every $d, n\geq 1$ and $G\leq\symg n$,
    every \svass{$(\symg d \wr G)$} $(\V, s, t)$ is $P(\norm{\V, s,t})$-\balrunable.
\end{lemma}
\begin{proof}
    Let $(\V, s, t)$ be an arbitrary \svass{$(\symg d \wr G)$}, and
    let $N := \norm{\V, s, t}$, $B_0 := 4N^2$, and $B := B_0  + 2N + 8N^3$.
    Suppose $\V$ has a run $s \trans{\pi} t$.
    We iteratively modify this run until it becomes $B$-\balrun.
    
    By definition $\norm s, \norm t < B_0 \leq B$, and therefore 
    the source configuration $s$ and the target one $t$ are $B$-\balrun.
    We make $\pi$ \balrun in $n$ stages.
    For each $j \in \setto n$, in an arbitrary order, we modify the run $\pi$ to make it \balrun inside the $j$th block, 
    i.e., to make all configurations of $\pi$ satisfy \eqref{eq:balconf} where $i,i'$ range over $\setto d$ but
    $j$ is fixed.
    As the modification does not affect the other blocks,
    the final run is therefore \balrun.

    From now on
    we concentrate on one stage, for a fixed $j\in\setto n$.    
    For a vector $\vr u\in\Z^{dn}$ and $i,i'\in \setto d$, we use the notation 
    $\diff{ii'}{\vr u} := \vr u(i, j) - \vr u(i',j)$ for the difference between the values of $\vr u$ 
    on coordinates $\pair i j$
    and $\pair {i'} j$.
    The stage amounts to iteratively modifying the run until it finally becomes \balrun inside the $j$th block.
    Each such individual modification we call a \emph{microstage}.
    Below we describe a single microstage, and argue that it decreases a certain nonnegative rank,
    which ensures termination of the whole stage.
     
     \para{Microstage}
    Suppose the run $\pi$ is not \balrun in the $j$th block, namely $\pi$ contains a configuration $p(\vr v)$
    such that $\diff{ii'}{\vr v}  > B$
    for some $i, i' \in \setto d$.
    Consider the longest infix
    \begin{align} \label{eq:inf}
    q(\vr w) \trans{\alpha} p(\vr v)\trans{\beta} q'(\vr w')
    \end{align}
    of  $\pi$, 
    such that all its  configurations $p'(\vr v')$
    satisfy $\diff{ii'}{\vr v'} \geq B_0$.
    Since one step may change the difference between the values of coordinates $(i, j)$ and $(i', j)$ 
    by at most $2N$, we notice that $\diff{ii'}{\vr w} < B_0 + 2N$, as otherwise $\alpha$ would be longer.
    Likewise, $\diff{ii'}{\vr w'} < B_0 + 2N$.
    In consequence $B - \diff{ii'}{\vr w} \geq 8N^3$ and $B - \diff{ii'}{\vr w'} \geq 8N^3$, and
    therefore $\alpha$ contains at least $4N^2$ steps that increase the difference
    between the values of coordinates $\pair i j$ and $\pair {i'} j$ and similarly $\beta$ contains
    at least $4N^2$ steps that decrease this difference.
    As there are at most $2N$ different values by which 
    a single step can increase (or decrease) this difference,
    by the pigeonhole principle $\alpha$ contains some $2N$ steps 
    that increase this difference by the same value $a$,
    and likewise $\beta$ contains some  $2N$ steps that decrease this difference by the same value $b$.
    We pick any $b$ among the $2N$ steps in $\alpha$, executing transitions 
    $t_1, \ldots, t_b$, say, and replace them by $\sigma(t_1), \ldots, \sigma(t_b)$,
    where $\sigma\in \wrprod{\symg d}{G}$ is a permutation that swaps $(i, j)$ and $(i', j)$.
    Likewise, we choose any $a$ among the $2N$ steps in $\beta$,
    executing transitions $t_1', \ldots, t_a'$, say,
    and replace them by $\sigma(t_1'), \ldots, \sigma(t_a')$.
    Importantly, this modification of the run does not change its target configuration.
    The only coordinate whose value decreases in some configurations along the run is $(i, j)$,
    but it never drops below zero since the decrease is never larger than $4N^2$, and $B_0 \geq 4N^2$.
    Therefore, the modification produces a run.
    
    We claim that the following nonnegative rank decreases as a result of the microstage
    (where the first sum ranges over configurations $r(\vr u)$ along the run):
    \begin{equation*}
        \sum_{r(\vr u)} \ \ \sum_{\ell < \ell' \in \setto d} \  \absv{\diff{\ell\ell'}{\vr u}},
    \end{equation*}
    on the basis of the following three observations.

    First, irrespective of the choice of $r(\vr u)$ along the run, the value $\absv{\diff{ii'}{\vr u}}$ never increases, 
    and decreases for some $r(\vr u)$.
    Indeed,
    for $p(\vr v)$ in \eqref{eq:inf}, the value $\absv{\diff{ii'}{\vr v}} = \absv{\diff{i'i}{\vr v}}$ decreases
    by $2ab > 0$ (but $\diff{ii'}{\vr v}$ remains positive).
    Moreover, for each configuration $r(\vr u)$ of $\alpha$ or $\beta$,
    $\absv{\diff{ii'}{\vr u}}$ decreases by a number between $0$ and $2ab$.

    Second, in all cases when the value $\absv{\diff{ii'}{\vr u}}$ changes (i.e., decreases) in the microstage,
    this is due to adding opposite values to $\vr u(i,j)$ and $\vr u(i',j)$,
    hence, $\vr u(i,j) + \vr u(i',j)$ remains unchanged.
    Therefore, for every $\ell \notin \{i, i'\}$ and every configuration $r(\vr u)$ in $\pi$, the value of
    $
        \absv{\diff{i\ell}{\vr u}} + \absv{\diff{i'\ell}{\vr u}}
    $
    does not increase.
    Indeed, if $\vr u(i,j) > \vr u(i', j) \geq \vr u(\ell, j)$ or 
    $\vr u(\ell, j) \geq \vr u(i,j) > \vr u(i', j)$ then the value
    $\absv{\diff{i\ell}{\vr u}} + \absv{\diff{i'\ell}{\vr u}}$ remains unchanged,
    and it might only decrease when $\vr u(i,j) > \vr u(\ell, j) > \vr u(i', j)$.
    We use here the fact that $\vr u(i,j)$ remains greater than $\vr u(i',j)$.
    Finally, for $\ell,\ell'\notin\{i,i'\}$, the value of $\diff{\ell\ell'}{\vr u}$ is preserved.
    In consequence, the rank of the whole run decreases, as required.
\end{proof}

\begin{theorem} \label{thm:small}
    \reachprob{(\symg d \wr \tg n)} 
    reduces in exponential time 
    to \reachprob n,
       for every $d,n\geq 1$.
\end{theorem}
\begin{proof}
    Let $\V = (Q, T)$ be \svass{$(\symg d \wr \tg n)$} and $s, t$ its configurations.
    Let $B:=P(N)$, where $P$ is the polynomial of
    Lemma \ref{lem:balancing} and $N=\norm{\V,s,t}$.
    Let $\B := \setto B \cup\set{0}$.
    Relying on Lemma \ref{lem:balancing} we construct an \parvass n $\V' = (Q', T')$,
    where $Q' = Q \times \B^{nd}$, that simulates $B$-\balrun runs of $\V$.
    Specifically, every $j$th block of $d$ coordinates of $\V$ is simulated by a single coordinate $j$ of $\V'$ whose
    value is invariantly equal to the smallest value in the $j$th block.
    The simulation is possible due to the auxiliary component $\B^{nd}$ of states that
    keeps the difference between the value of every coordinate $(i,j)$ of $\V$, 
    and the smallest value inside the $j$th block (to which $(i,j)$ belongs).
    The transitions $T'$ of $\V'$ simulate transitions of $T$, and accordingly update the auxiliary component of
    states. Formally, for a vector $\vr v \in \Z^n$, let $\overline{\vr v} \in \Z^{nd}$ denote 
    a vector such that $\overline{\vr v}(i,j) = \vr v(j)$ for $i \in [d]$, $j \in [n]$.
    Then
    \[
        T' = \setof{
            ((p, \vr b_p), \vr v, (q, \vr b_q)) \in Q' \times \Z^n \times Q'
        }{
            (p, \overline{\vr v} + \vr b_q - \vr b_p, q) \in T 
        }.   
    \]
    Ineed, this way
    there is a step from a configuration $(p, \vr b_p)(\vr u)$
    to a configuration $(q, \vr b_q)(\vr u')$ in $\V'$
    if and only if there is a step from a configuration
    $p(\overline{\vr u} + \vr b_p)$
    to a configuration $q(\overline{\vr u'} + \vr b_q)$ in $\V$.
    Observe that $T'$ is at most exponential in terms of the norm of $(\V, s, t)$.
    We claim that $\V$ has a $B$-\balrun run $s\trans{} t$ if and only if
    $\V'$ has a run $s' \trans{} t'$, where $s', t'$ are obtained from $s, t$ by computing the auxiliary component
    of states, and replacing the values of each block of coordinates by the smallest value in this block.
    The reduction is thus completed.    
\end{proof}

\section{Data \vass}\label{sec:data}

In this section we apply the approach of the previous section to \emph{data \vass}
\cite{DL09}, an extension of \vass
where every transition carries data from an infinite countable domain 
(e.g., natural numbers $\N$), and transitions can test equality of data values involved.
In consequence, transitions are invariant under all permutations of the data domain.

Data \vass are definable in the setting of symmetric \vass if we relax the
assumption that the set of coordinates is finite.
Using the combination of permutation groups of the previous section, 
we may define $n$-dimensional data \vass as 
\parvass{(\tg n \wr \symg \infty)}, for a combination 
of the trivial group $\tg n$ and a symmetric group $\symg \infty$
of \emph{infinite} countable degree.
The group $\symg \infty$ consists of all permutations of $\N$, 
the data domain, that fix all but finitely many elements of $\N$.
The set $\setto n \times \N$ 
of coordinates of a \parvass{(\tg n \wr \symg \infty)} is thus countable 
infinite.
In order to suitably adapt the model of symmetric \vass to infinitely many
coordinates, we still assume that every transition modifies only 
finitely many coordinates (and preserves all others).
The vector of every transition is thus nonzero on only finitely many coordinates.
Likewise, we assume that the source configuration (and hence also all reachable configurations)
is nonzero on only finitely many coordinates.
Finally, we assume orbit-finiteness of the set of transitions.
The set of transitions of a \parvass{(\tg n \wr \symg \infty)}
is necessarily infinite, due to 
the infinite size of the group $\symg \infty$, but it is finitely representable 
if we omit the zero coordinates in the orbit-representation.

\begin{example} \rm
The data \vass shown in Figure \ref{fig:datavass} has three orbits
of transitions:
\[
    \setof{(p, \ \vr v_a, \ q)}{a\in\N}, \quad 
    \setof{(q, \ \vr u_{ab}, \ r)}{a\neq b\in\N}, \quad 
    \setof{(r, \ \vr w_{ab}, \ q)}{a\neq b\in\N},
\]
where the vectors 
$\vr v_a$, $\vr u_{ab}$ and $\vr w_{ab} : \setto n \times \N \to \Z$ are zero on
all coordinates except the following ones:
\begin{align*}
 & \vr v_a(1, a)  = 2 && \vr u_{ab}(1, a) = \vr u_{ab}(1, b) = 1 && \vr w_{ab}(1, a)  = -2\\
 & \vr v_a(2, a)  = 3 && \vr u_{ab}(2, b) = -1 && \vr w_{ab}(2, b) = 1.
\end{align*}
For arbitrarily chosen data values $a\neq b \in \N$, 
let $\vr x, \vr y : \setto 2 \times \N \to \Z$ be vectors equal to zero on all
coordinates except the following ones: 
\begin{align*}
& \vr x(1, a) = 3 && \vr x(1, b) = 1 && \vr y(1, b) = 1 \\ 
&\vr x(2, a) = 2 && && \vr y(2, a) = 3.
\end{align*}
The configuration $r(\vr x)$ is reachable from $p(\vr 0)$
by firing the two transitions $(p, \vr v_a, q), (q, \vr u_{ba}, r)$, while
the configuration $q(\vr y)$ is not reachable from $p(\vr 0)$
as all configurations $q(\vr y')$ reachable from $p(\vr 0)$
necessarily satisfy the invariant $\sum_{i \in \N} \vr y'(1, i) = 2.$ 
\end{example}

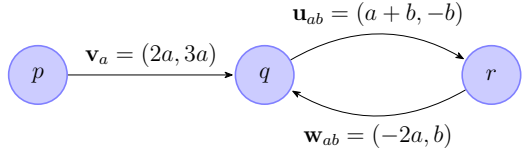
\begin{figure}
\begin{center}    
    \scalebox{0.75}{
    \begin{tikzpicture}[shorten >=1pt,node distance=4cm,on grid,>={Stealth[round]},
        every state/.style={draw=blue!50,thick,fill=blue!20},scale=0.25]
    
      \node[state]          (q_0)                      {$p$};
      \node[state]          (q_1) [right=of q_0] {$q$};
      \node[state]          (q_2) [right=of q_1] {$r$};
    
      \path[->] 
                (q_0) edge              node [above]  {$\vr v_a = (2a, 3a)$} (q_1)
                (q_1) edge [bend left]             node [above]  {$\vr u_{ab} = (a+b, -b)$} (q_2)
                (q_2) edge [bend left]            node [below]  {$\vr w_{ab} = (-2a, b)$} (q_1);
    \end{tikzpicture}
    }
    \caption{
        A 2-dimensional data \vass as a \parvass{(\tg 2 \wr \symg \infty)}. States are 
        $Q = \set{p, q, r}$, and transitions are given by orbit representatives.
        The data values $a\neq b\in \N$ are chosen arbitrarily. 
    }
    \label{fig:datavass}
\end{center}
\end{figure}

The reachability problem for data \vass is a long-standing open problem,
and is known to be decidable only for some restricted subclasses \cite{DL09,GL24,KL24}.
Using the techniques developed in this paper, we can prove decidability in 
yet another subclass of data \vass,
namely for data \vass invariant under independent permutations of data values 
in each of the dimensions $1\ldots n$.
In our setting, this subclass is definable as \parvass{(\symg \infty \wr \tg n)}, 
using another combination of 
the symmetric group of infinite degree and the trivial group $\tg n$.

The proof of Lemma \ref{lem:balancing} can be adapted to show that
\parvass{(\symg \infty \wr \tg n)} admit polynomially bounded runs.
Below, a run is called \emph{$B$-bounded} if all configurations occurring in it are
bounded by $B$.

\begin{lemma}\label{lem:inftybalancing}
    There is 
    a polynomial $P$ such that whenever a \parvass{(\symg \infty \wr \tg n)}
    $\V$ admits a run from $s$ to $t$, then it also admits a 
    $P(\norm{\V, s,t})$-bounded run from $s$ to $t$.
\end{lemma}
\begin{proof}
    First, we show that there is a polynomial $P$,
    such that every
    \parvass{(\symg \infty \wr \tg n)} $(\V, s, t)$ is $P(\norm{\V, s,t})$-\balrunable.
    It is enough to use the proof of Lemma \ref{lem:balancing},
    which works also for an infinite number of coordinates.
    Second, we observe that for \parvass{(\symg \infty \wr \tg n)},
    any $B$-\balrun run is also $B$-bounded.
    It follows from the fact, that in every run of \parvass{(\symg \infty \wr \tg n)}
    there is a counter (even infinitely many of them) with the value constantly equal to $0$.
    In any $B$-\balrun run the difference between the value of this counter
    and the value of any other counter cannot exceed $B$.
    This concludes the proof.
\end{proof}

\begin{theorem} \label{thm:inftysmall}
    \reachprob{(\symg \infty \wr \tg n)} is decidable, for every $n\in\N$.
\end{theorem}
\begin{proof}
    We provide a reduction from \reachprob{(\symg \infty \wr \tg n)} to the 
    (plain) \vass 
    reachability problem.
    
    Consider a given \parvass{(\symg \infty \wr \tg n)} $\V = (Q, T)$ and two 
    configurations $s,t$. 
    Let $B = P(\norm{\V, s,t})$, where $P$ is the polynomial of Lemma 
    \ref{lem:inftybalancing}. We thus know that if $\V$ has  a run 
    from $s$ to $t$ in $\V$, then it also has such a run that is $B$-bounded.
    
    For simplicity, assume that vectors in $s$ and $t$ are $\vr 0$.
    (Otherwise we would use extra counters in the constructed \vass to simulate
    those coordinates $S \subset_{fin} \N \times \setto n$ 
    on which the configurations $s$ or $t$ is nonzero.
    For simplicity, we prefer to omit these inessential details.)
    
    Let $d = n B$.
    We construct a \parvass {d} $\overline \V=(Q, \overline T)$ with the same control states as $\V$,
    and two its configurations $\overline s, \overline t$
    such that there is a $B$-bounded run from $s$ to $t$ in $\V$ 
    if and only if there is a run from $\overline s$ to $\overline t$ in $\overline{\V}$.
    Thus, $\overline{\V}$ is designed to simulate only $B$-bounded runs of $\V$.
    The main idea of the construction of $\overline{\V}$
    is to represent a configuration of $\V$ bounded by $B$
    by counting, separately for every $i\in\setto n$ and $b\in \setto B$,
    the number of coordinates $(j, i) \in \N\times\setto n$ equal to 
    $b$.
    More precisely, a configuration $p(\vr v)$ of $\V$ bounded by $B$ is 
    represented by the configuration $p(\overline{\vr v})$ of $\overline{\V}$,
    where:
    \[
        \overline{\vr v}((i-1)B + b) = 
            \size{\setof{(j, i) \in \N\times\setto n}{\vr v(j, i) = b}}
    \]
    for every $i\in\setto n$ and $b\in\setto B$.
    We define $\overline{\V} := (Q, \overline T)$, where 
    the transitions $\overline T$ are defined as 
    \[
        \overline{T} = 
        \setof{
            (p, \overline{\vr v + \vr w} - \overline{\vr v}, q)
        }{
            (p, \vr w, q) \in T
        },
    \]
    where vectors $\vr v + \vr w$ and $\vr v$ are
    implicitly assumed to be bounded by $B$.
    Indeed, this way there is a step from a configuration 
    $p(\vr v)$ to a configuration $q(\vr v + \vr w)$ in $\overline \V$
    if and only if there is a step from a configuration
    $p(\overline {\vr v})$ to a configuration $q(\overline{\vr v + \vr w})$ in $\overline \V$.
    As $T$ is assumed to be orbit-finite,
    the set $\overline T$ is finite (and exponential in terms of the 
    representation of $\V$).
    
    Correctness of the reduction follows by the construction.
    If $\V$ has 
    a run from $s$ to $t$ then relying on Lemma \ref{lem:inftybalancing},
    it also has a $B$-bounded run from $s$ to $t$,
    which yields a run of 
    $\overline \V$ from $\overline s$ to $\overline t$.
    Conversely, any run of $\overline \V$ from $\overline s$ to $\overline t$
    yields a run of $\V$ from $s$ to $t$.
\end{proof}

We remark that the decidability works in the uniform setting
when $n$ is part of the input.

\section{Other Groups}
\label{sec:other}

In this section we follow up Section \ref{sec:symg} by showing that alternating groups admit similar
complexity drop as symmetric groups.
We also demonstrate that transitivity alone is not enough to achieve any significant
complexity drop,
by showing that the complexity of the reachability problem in the case of cyclic groups is comparable to 
general \vass.

\para{Alternating group}

Let $\altg d \leq \symg d$ denote the alternating 
group of degree $d$, i.e., the group of all even
permutations of $\setto d$.
The group is generated by \emph{3-cycles}, permutations of the form
$
\sigma_{ijk} \ = \Big( {\small \begin{array}{c} i \ \ j \ \ k \\ j \  \ k \ \ i
\end{array}} \Big),
$
for pairwise distinct $i, j, k \in\setto d$.
When $d=2$ the alternating group is trivial, \parvass {\altg 2} = \parvass 2, and the reachability
problem is \pspace-complete \cite{Blondin15}.
In the sequel we thus consider degrees $d\geq 3$.

\begin{lemma}\label{lem:ad_bal}
    $\altg d$ is \fair, for $d \geq 3$.
\end{lemma}
$\altg d$ is obviously transitive,
hence we concentrate on proving the implication $(1) \implies (2)$
from the definition a fair group.
The argument elaborates on the proof of Lemma \ref{lem:sd_bal}, and proceeds in two steps.
First, assuming a run $\pi$ with one sufficiently large coordinate at the end,
by applying a permutation from $\altg d$ to a properly chosen suffix of $\pi$
we construct a run with \emph{two} large coordinates at the end.
Second, we further modify the run to increase all coordinates above a required threshold, similarly
to the proof of Lemma \ref{lem:sd_bal}.

\begin{proof}
    Consider an \svass{$\altg d$} $\V = (Q, T)$ and a configuration $s = q(\vr v)$.
    Let $N := \norm{\V}$ and $S := \norm{\V, s}$.
    Clearly, $S \geq N$.
    Let $R\in\N$ be any constant, and
    $B := 2(N+1)(S+R)$.
    In order to prove that $\altg d$ is \fair for the polynomial $P(x) = 3d \cdot 2(x+1)+1$, we 
    assume that $\V$ has a run $\pi \colon q(\vr v) \trans{} p(\vr w)$ with $\norm{\vr w} > B\cdot 3d + N$,
    and construct a run $\pi'' \colon q(\vr v) \trans{} p''(\vr w'')$ whose target vector satisfies 
    $\vr w'' \geq \const{S+R+1}$. 
    
\smallskip

    \Wlog~we assume that $\vr w(d) > B\cdot 3d + N$.
    We proceed in two steps.
    In the first step we transform $\pi$ into a run $\pi' \colon q(\vr v) \to p'(\vr w')$,
    where $\vr w'(d) > B\cdot d$ and $\vr w'(i) > B\cdot d$ for some $i \neq d$.
    Suppose $\pi$ does not satisfy the required condition, namely $\vr w(i) \leq B\cdot d$ for $i\in \setto{d-1}$.
    We split $\pi$ into:
    \[
    q(\vr v) \trans{\rho} q'(\vr v') \qquad\qquad
    q'(\vr v') \trans{\tau} p(\vr w),
    \]
    where $\tau$ is a suffix of $\pi$ (not necessarily unique)
    in which the value of coordinate $d$
    is always greater than $B \cdot 2d$, but $\vr v'(d)$ is the lowest possible.
    Hence, the value of coordinate $d$ never drops below $\vr v'(d)$ along $\tau$,
    and $\vr v'(d) \leq B\cdot 2d + N$, as otherwise we could choose a lower $\vr v'(d)$.
    Let $i, j \in \setto{d-1}$ be any coordinates with $\vr v'(i) \leq \vr v'(j)$.
    If $\vr v'(j) > B\cdot d$ we take $\pi' = \rho$.
    Otherwise, we take
    \[
    \pi'  \ := \ \rho ; \ \sigma_{ijd}(\tau),
    \]
    i.e., for each step in $\tau$ we move the effect on coordinate $i$ to $j$,
    the effect on coordinate $j$ to $d$, and the effect on coordinate $d$ to  $i$.
    We need however to argue that $\pi'$ is indeed a run, i.e., it never drops below 0 on coordinates
    $i, j, d$.
    As $\vr v'(j) \leq B\cdot d$, 
    the effect of every prefix of $\tau$ on coordinate $j$
    is $\geq -B\cdot d$.
    Therefore, since $\vr v'(d) > B\cdot 2d$,
    the coordinate $d$ is always greater than $B \cdot 2d - B\cdot d = B\cdot d$ 
    along $\sigma_{ijd}(\tau)$.
    The coordinate $i$  never drops below $\vr v'(i)$ along $\sigma_{ijd}(\tau)$,
    because we have chosen $\tau$ so that the coordinate $d$ never drops below $\vr v'(d)$ along $\tau$. 
    Finally, coordinate $j$ never drops below $0$ along $\sigma_{ijd}(\tau)$ as $\vr v'(i) \leq \vr v'(j)$.
    We have thus demonstrated that $\pi'$ is a run.

    Let $p(\vr w')$ be the final configuration of $\pi'$.
    Since $\vr v'(d) \leq B\cdot 2d + N$ and $\vr w(d) > B\cdot 3d + N$,
    the coordinate $d$ is increased by $\tau$ by more than $B\cdot d$, 
    and hence the coordinate $i$ is increased by $\sigma_{ijd}(\tau)$ by more than $B\cdot d$, 
    which implies $\vr w'(i) > B\cdot d$.
    Summing up, we have $\vr w'(i), \vr w'(d) > B\cdot d$, as required.

\smallskip

    In the second step we will modify the run $\pi'$
    similarly as in the proof of Lemma \ref{lem:sd_bal}.
    \Wlog~assume that $\vr w'(d-1), \vr w'(d) > B \cdot d$.
    For $i = 0, 1, \ldots, d-2$ we inductively construct runs $\pi_i \colon q(\vr v) \to p'(\vr w_i)$ satisfying
    \begin{align}\label{eq:awi}
    \vr w_i(d), \vr w_i(d-1) > B \cdot (d-i) \qquad \qquad
    \vr w_i(j) > S + R \ \ (\text{for all }j \geq i).
    \end{align}
    We start with $\pi_0 := \pi'$.
    Assuming $\pi_{i-1}$ has already been constructed, we construct $\pi_{i}$ as follows.
    If $\vr w_{i-1}(i) > S + R$ then we take $\pi_{i} := \pi_{i-1}$.
    Otherwise, we split $\pi_{i-1}$ into 
    $\pi_{i-1} \ = \ \alpha ; \beta,$
    \[
    q(\vr v) \trans{\alpha} r(\vr u)
    \qquad\qquad
    r(\vr u) \trans{\beta} p(\vr w_{i-1}),
    \]
    where $\beta$ is the longest suffix of the run
    in which the values of the coordinates $d$ and $d-1$ are both greater than $B\cdot (d-i)$.
    One of $\vr u(d-1)$, $\vr u(d)$ is $<B\cdot(d-i)+N$, as otherwise $\beta$ would be longer.
    \Wlog~$\vr u(d) < B\cdot(d-i)+N$.
    Let $\Delta = \vr w_{i-1} - \vr u$ be the difference between the target and source of $\beta$.
    We observe that the value of coordinate $d$
    is increased by $\beta$ by at least
    \[
     \Delta(d) \ > \ B\cdot(d-i+1) - B\cdot(d-i)-N = 2(N+1)(S+R) - N \geq 2N(S+R) + S + R,
    \]
    while the value of coordinate $i$ is increased by at most 
    \[
    \Delta(i) \ \leq \ \vr w_{i-1}(i) \ \leq \ S+R.
    \]
    Therefore, $\Delta(d) - \Delta(i) \geq S+R+1$.
    In consequence, $\beta$ contains some $k \leq S+R+1$ steps
    \[
    q_1(\vr u_1) \trans{t_1} q'_1(\vr u'_1) \quad \ldots \quad
    q_k(\vr u_k) \trans{t_k} q'_k(\vr u'_k)
    \]
    induced by transitions $t_1 = (q_1, \vr e_1, q'_1), \ldots, t_k = (q_k, \vr e_k, q'_k)$
    with effects 
    \[
    \vr e_1 \ = \ \vr u'_1 - \vr u_1 \qquad \ldots \qquad 
    \vr e_{k} \ = \ \vr u'_k - \vr u_k
    \]
    such that $\vr e_\ell(d) > \vr e_\ell(i)$ for $\ell \in \setto {k}$, that is, each of the steps increases 
    the difference between 
    the values of coordinate $d$ and $i$, and moreover
    \begin{align} \label{eq:aNR1}
    \vr e_1(d) + \ldots + \vr e_k(d) \ \geq \ \vr e_1(i) + \ldots + \vr e_k(i) + S+R+1,
    \end{align}
    that is, the steps jointly increase the difference between the values of coordinate $d$ and $i$ 
    by ay least $S+R+1$.
    We define a run $\pi_{i} \colon q(\vr v) \runarrow p'(\vr w_{i})$
    by replacing each transition $t_\ell$ in $\pi_{i-1}$ with the transition $\sigma_{i(d-1)d}(t_\ell)$.
    As one step changes the difference between coordinates by at most $2N$,
    the coordinate $d$ never drops below 
    \[
    B\cdot(d-i) - 2N \cdot k \ \geq \ B\cdot(d-i) - 2N(S+R+1) \ \ge \ 0
    \] 
    in $\pi_{i}$, 
    and $\vr w_{i}(d) > B\cdot(d-i+1) - 2N\cdot k \geq B\cdot(d-i+1) - 2N(S+R+1) \ge B\cdot(d-i)$.
    Similarly, we prove that the coordinate $d-1$ never drops below $0$ 
    and $\vr w_{i}(d-1) > B\cdot(d-i)$.
    Furthermore, due to the inequality \eqref{eq:aNR1}, $\vr w_{i}(i) > S + R$,
    because the coordinate $i$ was increased
    by at least $S+R+1$.
    Other coordinates are not affected, i.e., their values are the same in $\pi_{i-1}$ and $\pi_{i}$.
    Therefore, the run $\pi_{i}$ satisfies the inequalities \eqref{eq:awi}, as required.

    Finally, we arrive at a run $\pi'' = \pi_{d-2} \colon q(\vr v) \runarrow p'(\vr w_{d-2})$
    with $\vr w_{d-2}(i) \geq S+R+1$ for every $i \in \setto d$, and this completes the proof.
\end{proof}

\begin{theorem}\label{thm:An}
    \reachprob {\altg d} is \pspace-complete, for every $d \geq 3$.
\end{theorem}

\begin{proof}
Membership in \pspace follows immediately by Lemmas \ref{lem:pspace} and \ref{lem:ad_bal},
while \pspace-hardness follows by Lemma \ref{lem:hardness}.
\end{proof}

Similarly to Theorem \ref{thm:Sn}, the \pspace upper bound holds even in the uniform\
setting, where the degree $d$ is part of the input.
We note that Theorem \ref{thm:An} subsumes Theorem \ref{thm:Sn}, as 
\parvass{\symg d} reduce polynomially to \parvass{\altg d}.
Indeed, any \parvass{\symg d} $\V = (Q, T)$ is automatically an \parvass{\altg d} 
(as $\altg d \leq \symg d$) whose representation is obtained from a representation
of $\V$ by adding to the set of transitions $T$
the image of each transition under every single transposition.
Thus, the size of the representation increases only polynomially.

\para{Cyclic groups}

Let $\cycg n$ denote the cyclic group of degree $n$, i.e., the group 
$\cycg n = \setof{\sigma^i}{i\in\{0, 1, \ldots, n-1\}}$
generated by the single cyclic shift
$
\sigma \ = \Big( {\small \begin{array}{c} 1 \ \ 2 \ \ \ldots \ \ n \\ 2 \  \ 3 \ \  \ldots \ \ 1
\end{array}} \Big).
$

\begin{theorem}\label{thm:zn}
    \reachprob{n}
    reduces in polynomial time to \reachprob{(\cycg{2n+8})},  for every $n\geq 1$.
\end{theorem}
In the proof, an \parvass n 
is simulated using $n$ consecutive counters of a \parvass {\cycg {2n+8}}.
The remaining $n+8$ counters are used
to enforce that only one transition from each $\cycg {2n+8}$-orbit can be applied,
which assures faithfulness of the simulation.
\begin{proof}
    Given an \parvass n $\V = (Q, T)$, we construct
    a \parvass{(Z_{2n+8})} $\V' = (Q', T')$ 
    with states $Q' = Q \cup \setof{q'}{q \in Q}$
    and transitions $T'$.
    The idea is to simulate 
    a configuration $q(\vr v)$ of $\V$  
    by the configuration $\overline {q(\vr v)} = q(\overline {\vr v})$ of $\V'$,
    where
    \begin{align*}
    \overline {\vr v}(2) \ = \ & \overline {\vr v}(n+3) \ = \  1 &
    \overline{\vr v} (i+2) \ = \ & \vr v(i) \qquad (\text{for } i \in [n])\; \\
    && \overline {\vr v}(i)  \ = \ & 0 \qquad\quad (\text{for other } i).
    \end{align*}
    We define the orbit representatives $\widetilde T\subseteq T'$ as follows.
    For every $t = (p, \vr v, q) \in T$
    we put to $\widetilde T$ a transition $t' = (p, \vr v', q')$, where
    \begin{align*}
    \vr v'(2) \ = \ & \vr v'(n+3) \ = \  -1 &
    \vr v'(i+2) \ = \ & \vr v(i) \qquad (\text{for } i \in [n]) \\
    \vr v'(1) \ =  \ & \vr v'(n+4) \ = \ 1 &
    \vr v'(i) \  = \ & 0 \qquad\quad (\text{for other } i).
    \end{align*}
    %
    Furthermore,
    for every $q \in Q$
    we put to $\widetilde T$ a transition $u_q = (q', \vr v'', q)$, where
    \begin{align*}
    \vr v''(2) \ = \ & \vr v''(n+3) \ = \  1 \\
    \vr v''(1) \ =  \ & \vr v''(n+4) \ = \ -1 &
    \vr v''(i) \  = \ & 0 \qquad\quad (\text{for other } i).
    \end{align*}
    
     We demonstrate correctness of the reduction.
     In one direction,
     every step $c \trans{t} c'$ in $\V$, where $t = (p, \vr v, q) \in T$,
     is simulated in $\V'$ by a two-transition run 
     \[
     \overline c \trans{\alpha_t} \overline{c'}
     \qquad \text{ where }
     \alpha_t \ = \ t'; \ u_q,
     \] 
     and therefore $c\trans{} c'$ in $\V$ implies $\overline c \trans{} \overline{c'}$ in $\V'$. 
    Conversely, suppose $\overline c \trans{} \overline{c'}$ in $\V'$.
    By the choice of states in the definition of  transitions $T'$, 
    the run may be split into two-transition segments of the form
    \[
    b\trans{\alpha} b'
    \qquad \text{ where }
    \alpha \ = \ \sigma^k(t'); \ \sigma^\ell(u_q),
    \]
    for some $t = (p, \vr v, q) \in T$.
    Suppose $b = p(\overline {\vr w})$.
    Since ${\vr v'}(2) = \vr v'(n+3) = -1$, i.e.,
    $\vr v'$ contains a pair of negative numbers separated by window of length $n$, 
    and $\overline{\vr w} = 0$ for all $i \notin \{2, \ldots, n+3\}$, 
    we deduce that $k=0$, i.e.,
    $\sigma^k(t') = t' \in \widetilde T$.
    Likewise, we deduce that $\ell = 0$, i.e., $\sigma^\ell(u_q) = u_q \in \widetilde T$.
    Therefore, $\alpha = \alpha_t$, and hence $b' = q(\overline{\vr w + \vr v})$.
    Reasoning in this way for all two-transition segments, we deduce 
    that $\overline c \trans{} \overline{c'}$ in $\V'$ implies $c\trans{} c'$ in $\V$, as required.
\end{proof}


\section{Final Remarks}
\label{sec:final}

In this paper we have investigated the impact of symmetry in \vass on the complexity
of the reachability problem.
The permutation groups considered in this paper split clearly into two groups.
On one side there are
`easy' groups, which includes both symmetric and alternating groups featuring \pspace-completeness,
but also the combinations $\wrprod{\symg d}{\tg n}$ of symmetric groups and trivial ones,
where the complexity is independent of the degree $d$ of the symmetric group
(dependence on the degree $n$ of the trivial group is unavoidable).
On the other side there are
`hard' groups, which includes both trivial and cyclic groups, which do not differ significantly in complexity, 
but also the dual combinations $\wrprod{\tg n}{\symg d}$ of symmetric groups and trivial ones,
where the complexity is dependent on both $d$ and $n$.
We end up with a question left for the further research:
can one classify \emph{all} finite permutation groups into `easy' and `hard' ones?

\smallskip

The model of data \vass is definable in our setting as \parvass{(\wrprod{\tg n}{\symg \infty})}.
Therefore,
`hardness' of $\wrprod{\tg n}{\symg d}$ may suggest that the reachability problem
for data \vass is harder than for plain \vass.
On the other hand, we prove that the reachability problem for the subclass of data \vass
invariant under independent
data permutations in each dimension
(definable as  \parvass{(\wrprod{\symg \infty}{\tg n})} in our setting)
is not significantly harder than for plain \vass
(and, in particular, decidable).
This establishes an interesting connection between symmetric \vass and data \vass.

\bibliographystyle{elsarticle-num} 
\bibliography{bib,pn-bib}






\end{document}